\documentclass[reqno,12pt,letterpaper]{amsart}
\usepackage{verbatim}

\DeclareFontFamily{U}{mathx}{\hyphenchar\font45}
\DeclareFontShape{U}{mathx}{m}{n}{
      <5> <6> <7> <8> <9> <10>
      <10.95> <12> <14.4> <17.28> <20.74> <24.88>
      mathx10
      }{}
\DeclareSymbolFont{mathx}{U}{mathx}{m}{n}
\DeclareFontSubstitution{U}{mathx}{m}{n}
\DeclareMathAccent{\widecheck}{0}{mathx}{"71}
\DeclareMathAccent{\wideparen}{0}{mathx}{"75}

\usepackage{amsmath,amssymb,amsthm,graphicx,mathrsfs,url}
\usepackage[usenames,dvipsnames]{color}
\usepackage[colorlinks=true,linkcolor=Red,citecolor=Green]{hyperref}
\usepackage{floatrow}
\usepackage[usenames,dvipsnames]{xcolor}
\usepackage{tikz} 
\usepackage{tikz-cd} 
\usepackage{slashed}
\usepackage{graphicx}
\usepackage{color}
\usepackage{amsmath,amsfonts}
\usepackage{calrsfs}
\usepackage{amsmath,environ}
    \usetikzlibrary{decorations.pathreplacing}
\usepackage{autonum}
\usepackage{amssymb}
\usepackage{amsmath, amsthm}
\usepackage{dsfont}
\usepackage{fancybox}
\usepackage{amssymb}

\usetikzlibrary{decorations.pathreplacing}

\DeclareMathAlphabet{\pazocal}{OMS}{zplm}{m}{n}

\newcommand{\Di}{\slashed{D}}

\newcommand{\Range}{{\operatorname{Range}}}

\newcommand{\hphi}{{\widehat{\phi}}}
\newcommand{\hvp}{{\widehat{\vp}}}
\newcommand{\haz}{{\widehat{\az}}}

\newcommand{\VVV}{{\mathscr{V}}}

\newcommand{\R}{\mathbb{R}}

\newcommand{\N}{\mathbb{N}}

\newcommand{\p}{{\partial}}

\newcommand{\matrice}[1]{\left[ \begin{matrix}
#1
\end{matrix} \right]}
\newcommand{\Det}{{\operatorname{Det}}}

\newcommand{\epsi}{\varepsilon}
\newcommand{\sgn}{{\operatorname{sgn}}}
 
\newcommand{\Tr}{{\operatorname{Tr}}}

\newcommand{\lr}[1]{\langle #1 \rangle}

\newcommand{\FF}{\pazocal{F}}

\newcommand{\tsigma}{{\tilde{\sigma}}}
\newcommand{\txi}{{\tilde{\xi}}}
\newcommand{\tb}{{\tilde{b}}}

\newcommand{\GG}{\pazocal{G}}

\newcommand{\supp}{\mathrm{supp}}

\newcommand{\BB}{\mathcal{B}}

\newcommand{\Z}{\mathbb{Z}}

\newcommand{\Bb}{\mathbb{B}}

\newcommand{\Ll}{\mathbb{L}}

\newcommand{\CC}{\pazocal{C}}

\newcommand{\Id}{{\operatorname{Id}}}
\newcommand{\VV}{{\pazocal{V}}}

\newcommand{\EE}{\pazocal{E}}

\newcommand{\MM}{\pazocal{M}}

\newcommand{\LL}{{\pazocal{L}}}
\newcommand{\HH}{\pazocal{H}}
\newcommand{\tHH}{\widetilde{\HH}{}}
\newcommand{\SSS}{{\pazocal{S}}}
\newcommand{\ZZZ}{{\pazocal{Z}}}
\newcommand{\Hh}{{\mathbb{H}}}
\newcommand{\vp}{{\varphi}}

\newcommand{\systeme}[1]{\left\{ \begin{matrix} #1 \end{matrix} \right.}

\newcommand{\tH}{{\widetilde{H}}}

\newcommand{\C}{\mathbb{C}}

\newcommand{\az}{\alpha}

\newcommand{\rk}{{\operatorname{rk}}}

\newcommand{\Mm}{\mathbb{M}}

\newcommand{\Tt}{{\mathbb{T}}}

\newcommand{\RRR}{\pazocal{R}}

\newcommand{\de}{ \ \mathrel{\stackrel{\makebox[0pt]{\mbox{\normalfont\tiny def}}}{=}} \ }

\numberwithin{equation}{section}

\title[Ubiquity of conical points in topological insulators]{Ubiquity of conical points in topological insulators}
\author{Alexis Drouot}

\NewEnviron{equations}{%
\begin{equation}\begin{gathered}
  \BODY
\end{gathered}\end{equation}
}
\newtheorem{thm}{Theorem}
\newtheorem{conj}{Conjecture}
\newtheorem{cor}{Corollary}
\newtheorem{defi}{Definition}
\newtheorem{lem}{Lemma}[section]

\newtheorem{theorem}[thm]{Theorem}
\theoremstyle{definition}
\newtheorem{rmk}{Remark}[section]

\begin{document}

\maketitle

\begin{abstract} We show that generically, the degeneracies of a family of Hermitian matrices depending on three parameters have a \textit{conical} structure. Our result applies to the study of topological phases of matter. It implies that adiabatic deformations of two-dimensional topological insulators come generically with \textit{Dirac-like} propagating currents, whose total conductivity equals the chiral number of conical points.
\end{abstract}

\section{Introduction}

Let $\EE$ be the space of $N \times N$ Hermitian matrices; $\EE^* \subset \EE$ consisting of matrices with simple eigenvalues; and $\Tt^2$ be a two-dimensional torus. Given $H_0$ and $H_1$ in $C^\infty(\Tt^2,\EE^*)$, is there a path from $H_0$ to $H_1$, that remains in $C^\infty(\Tt^2,\EE^*)$? 

In general, the response is no: there is a topological obstruction, related to the eigenbundles of $H_0$ and $H_1$. When this obstruction is present, \textit{any} path from $H_0$ to $H_1$ acquires degenerate eigenvalues. In this paper, we explore the shape of these crossings. We show that generically, they exhibit a conical structure. 

This result has a counterpart in the theory of topological phases of matter. When two topologically distinct insulators are adiabatically connected, it implies that generically:
\begin{itemize}
\item Finitely many channels supporting chiral currents appear;
\item Up to large times, these currents follow a Dirac equation and are concentrated (in phase-space) along conical eigenvalue crossings;
\item The chiral number of currents  equals the Chern number difference.
\end{itemize}
This establishes a quantitative link between (a) asymmetric currents; (b) eigenvalue crossings; and (c) the bulk-edge correspondence.

\subsection{Genericity of conical points}\label{sec:1.1} We first state our result in a form that applies to topological phases of matter. We postpone the general statement to \S\ref{sec:1.4}. 

Let $\Tt^2 = \R^2/(2\pi\Z)^2$ and $H_0, H_1$ be two elements of $C^\infty(\Tt^2,\EE)$, with eigenvalues $\lambda_1\big( H_j(\xi) \big) \leq \dots \leq \lambda_N\big( H_j(\xi) \big)$, repeated according to multiplicity. We assume that for some $n \in [0,N-1]$ and all $\xi \in \Tt^2$,
\begin{equation}\label{eq:1e}
\lambda_n\big(H_0(\xi)\big) < \lambda_{n+1}\big(H_0(\xi)\big), \ \ \ \  \lambda_n\big(H_1(\xi)\big) < \lambda_{n+1}\big(H_1(\xi)\big).
\end{equation}
Let $\LL$ be the set of smooth homotopies from $H_0$ to $H_1$:
\begin{equation}
\LL \de \Big\{ H \in C^\infty([0,1] \times \Tt^2, \EE), \ H(0,\cdot) = H_0, \ H(1,\cdot) = H_1\Big\}.
\end{equation}

\begin{defi}\label{def:1} If $H \in \LL$, we say that $\lambda_n(H)$ and $\lambda_{n+1}(H)$ cross (or degenerate) at $\zeta_0 = (s_0,\xi_0) \in [0,1] \times \Tt^2$ if  $\lambda_n\big(H(\zeta_0)\big) = \lambda_{n+1}\big(H(\zeta_0)\big)$. \\
\indent We say that $\lambda_n(H)$ and $\lambda_{n+1}(H)$ cross conically  if $\lambda_n\big( H(\zeta_0) \big)$ has multiplicity precisely two; and if there exist $a_0 \in \R^3$ and $S_0 \in M_3(\R)$ invertible such that 
\begin{equation}\label{eq:1f}
\systeme{
\lambda_n\big(H(\zeta_0+\epsi)\big) & = & \lambda_n\big( H(\zeta_0) \big) + \lr{a_0,\epsi} - \| S_0 \epsi \| + o(\epsi)
\\
\lambda_{n+1}\big(H(\zeta_0+\epsi)\big) & = & \lambda_n\big( H(\zeta_0) \big) + \lr{a_0,\epsi} + \| S_0 \epsi \| + o(\epsi)}, \ \ \ \ \text{$\epsi \in \R^3$ \ small. } 
\end{equation}
\end{defi}

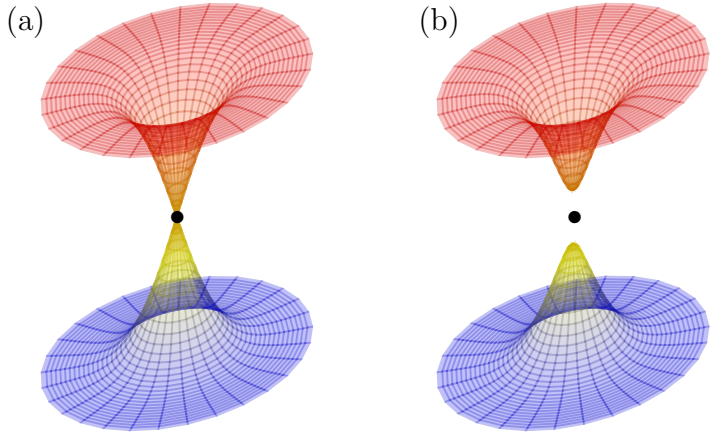
\begin{figure}
\floatbox[{\capbeside\thisfloatsetup{capbesideposition={right,center},capbesidewidth=3in}}]{figure}[\FBwidth]
{\caption{(a) Eigenvalue surfaces of $H(s_0,\cdot)$ near a conical point $(s_0,\xi_0)$ of $H$. They intersect at the vertex of a (non-isotropic) cone. (b) Eigenvalue surfaces of $H(s,\cdot)$ for $s \neq s_0$ near $s_0$. They no longer touch.}\label{fig:2}}
{\begin{tikzpicture}
  \node at (-3,0) {\includegraphics[height=2in,trim={1.4cm 0cm 2cm 1.45cm}]{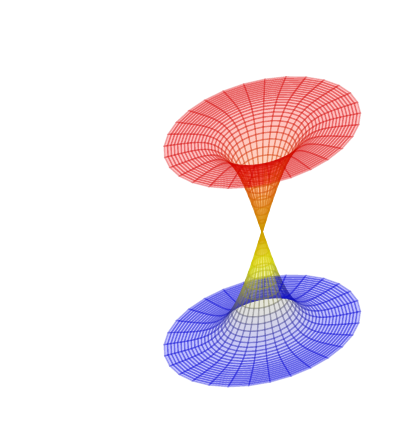}};

      \node at (1,0) {\includegraphics[height=2in,trim={-.5cm 0cm 1.5cm 1.45cm}]{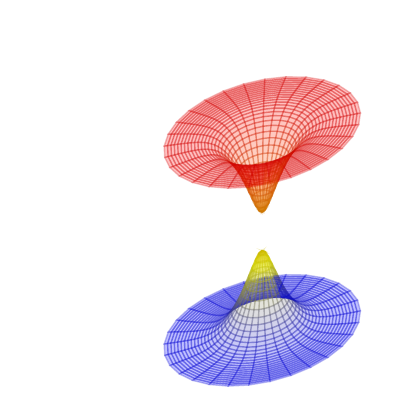}};
%\draw[black,thin,dotted] (-3.5,-4.5) grid (4.5,4.5);  
  
   \node at (-1.48,.9) {$\bullet$};

   \node at (3.8,.9) {$\bullet$};

\node at (-3.5,3.5) {(a)};
   \node at (2,3.5) {(b)};
   
\node[white] at (4.6,1) {(b)};
   \end{tikzpicture}}
\end{figure}

Conical degeneracies correspond to tilted cones in the graphs of eigenvalues -- see Figure \ref{fig:2}.
In particular, conical crossings of $\lambda_n(H)$ and $\lambda_{n+1}(H)$ are isolated. At first, we could think that they are rare among degeneracies: a non-empty intersection of two surfaces is in general a curve (rather than a point). Nonetheless:

\begin{thm}\label{thm:1} If $H_0$ and $H_1$ are elements of $C^\infty(\Tt^2,\EE)$ satisfying \eqref{eq:1e}, then
\begin{equation}
\Ll  = \big\{ H \in \LL : \ \text{all crossings of $\lambda_n(H)$ and $\lambda_{n+1}(H)$ are conical } \big\}
\end{equation}
is a dense open subset of $\LL$.
\end{thm}

The natural topology on $\LL$ is that induced by $C^\infty\big([0,1] \times \Tt^2,\EE\big)$, see \S\ref{sec:1.6}. Results at lower regularity are also possible; our techniques typically require $C^2$.

As a simple consequence of Theorem \ref{thm:1}, for generic $H \in \LL$, $\lambda_n(H)$ and $\lambda_{n+1}(H)$ cross at only finitely many points: conical crossings are isolated. Under a topological condition on $H_0$ and $H_1$, crossings must nonetheless arise. Indeed,  \eqref{eq:1e} allows us to define a rank-$n$ vector bundle $\VVV_0$ over $\Tt^2$: the fibers are 
\begin{equation}
\VVV_0(\xi) = 
\bigoplus_{j=1}^n \ker\Big( H_0(\xi) - \lambda_j\big( H_0(\xi) \big) \Big), \ \ \ \ \xi \in \Tt^2.
\end{equation}
We can also define $\VVV_1$, associated to $H_1$: only \eqref{eq:1e} is necessary to construct such vector bundles. Hence, if there is a homotopy between $H_0$ and $H_1$ that maintains \eqref{eq:1e}, then there are smooth vector bundles $\VVV_s \rightarrow \Tt^2$, $s\in [0,1]$, interpolating between $\VVV_0$ and $\VVV_1$. In particular, $\VVV_0$ and $\VVV_1$ would be topologically equivalent. 

This restriction can be measured via the Chern number -- the vector bundle analog of the Euler characteristic. This number can take any integer value, even in the context of eigenbundles -- see the appendix in  \cite{D19c} -- and characterizes the topology when the basis is a two-torus -- see e.g. \cite{P07,M17}. Thus, $\VVV_0$ and $\VVV_1$ are topologically equivalent if and only if $c_1(\VVV_0) = c_1(\VVV_1)$. In particular, if $H_0, H_1 \in C^\infty(\MM,\EE)$ satisfy \eqref{eq:1e} and $c_1(\VVV_0) \neq c_1(\VVV_1)$, then \textit{any} homotopy between $H_0$ and $H_1$ admits degeneracies. These, according to Theorem \ref{thm:1}, are generically all conical -- see Figure \ref{fig:1}. 

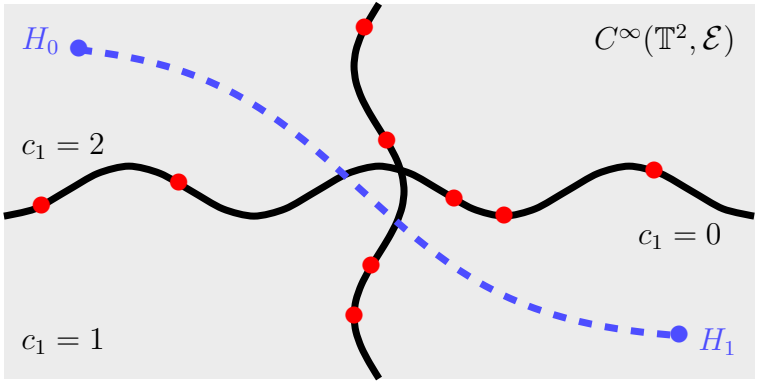
\begin{figure}
\floatbox[{\capbeside\thisfloatsetup{capbesideposition={right,center},capbesidewidth=2.8in}}]{figure}[\FBwidth]
{\hspace{-1cm}\caption{For each $n \in [1,N-1]$, $C^\infty(\Tt^2,\EE)$ splits in components distinguished by Chern numbers. If $H_0$ and $H_1$ lie in different components, a path joining $H_0$ to $H_1$ (blue) acquire crossings. Non-conical-type degeneracies (red) are rare in $C^\infty(\Tt^2,\EE)$. }\label{fig:1}}
{\begin{tikzpicture}
  \node at (0,0) {\includegraphics[height=2in]{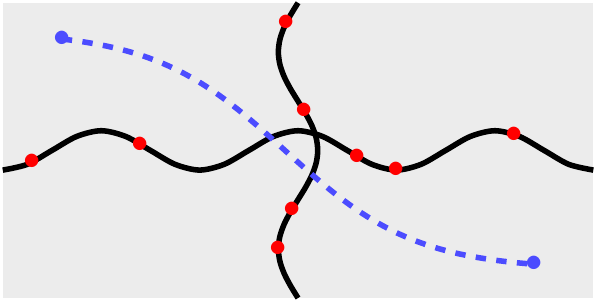}};

\node at (-4.2,-2) {$c_1 = 1$} ;
\node at (-4.2,.6) {$c_1 = 2$} ;
\node at (4,-.6) {$c_1 = 0$} ;
\node[blue!70] at (-4.5,2) {$H_0$} ;
\node[blue!70] at (4.5,-2) {$H_1$} ;
\node at (3.8,2) {$C^\infty(\Tt^2,\EE)$} ;

   \end{tikzpicture}}
\end{figure}

\subsection{Connection with topological phases of matter}\label{sec:1.2}  We review tight-binding, translation-invariant models of insulators at an energy $\lambda_0 \in \R$. These systems are represented by selfadjoint Hamiltonians $\HH_0 : \ell^2(\Z^2,\C^N) \rightarrow \ell^2(\Z^2,\C^N)$ with:
\begin{equation}\label{eq:1a}
[\HH_0,T_j] = 0, \ \ \ \ (T_j \psi)_m = \psi_{m+e_j}; \ \ \ \ \text{and} \ \ 
\lambda_0 \notin \sigma(\HH_0).
\end{equation}
In \eqref{eq:1a}, $\sigma(\HH_0)$ denotes the $\ell^2(\Z^2,\C^N)$-spectrum of $\HH_0$. Physically, $\lambda_0 \notin \sigma(\HH_0)$ means that there is no plane-wave propagation at energy $\lambda_0$.

Thanks to \eqref{eq:1a} and $[T_1,T_2] = 0$, we can diagonalize $\HH_0$, $T_1$ and $T_2$ simultaneously. The eigenvalues of $T_j$ are $e^{i\xi_j}$, $\xi_j \in \Tt^1 = \R/(2\pi \Z)$. Joint eigenspaces of $T_1$ and $T_2$ canonically identify with $\C^N$:
\begin{equation}
\bigcap_{j=1}^2 \ker\left(T_j-e^{i\xi_j}\right) = \left\{ \big(e^{i\xi m} \psi_0\big)_{m \in \Z^2} \ : \ \psi_0 \in \C^N\right\} , \ \ \ \ \xi =(\xi_1,\xi_2) \in \Tt^2.
\end{equation}
Thus, the analysis of $\HH_0$ reduces to that of its Bloch transform: the $\Tt^2$-parametrized family of $N \times N$ Hermitian matrices
\begin{equation}\label{eq:3b}
H_0(\xi) = e^{-i\xi m} \cdot \HH_0 \cdot e^{i\xi m}, \ \ \ \ \xi \in \Tt^2.
\end{equation}
The insulating condition $\lambda_0 \notin \sigma(\HH_0)$ and the spectral decomposition of $\HH_0$ into $\{ H_0(\xi) \}_{\xi \in \Tt^2}$ imply that $\lambda_0$ is never in $\sigma\big( H_0(\xi) \big)$. Thus, $H_0$ satisfies \eqref{eq:1e}. 

A standard question in topological phases of matter is whether two materials can be deformed to each other while maintaining their electronic properties. If $\HH_1$ is another insulator at energy $\lambda_0$, with associated vector bundle $\VVV_1$ of rank $n$, then $H_1$ also satisfies \eqref{eq:1e}. As explained in \S\ref{sec:1.1}, if $c_1(\VVV_0) \neq c_1(\VVV_1)$, then there are no path $\{\HH_s\}_{s \in [0,1]}$ connecting $\HH_0$ and $\HH_1$ while maintaining \eqref{eq:1a}. Physically, two topologically distinct insulators cannot be deformed to one another without passing by a conductor.

 Theorem \ref{thm:1} explains quantitatively this failure. Generically, conical crossings arise as one transitions from $\HH_0$ to $\HH_1$. The quantity $c_1(\VVV_1) - c_1(\VVV_0)$ is fundamental in the analysis of interface effects between topological insulators; see e.g. \cite{RH08,B19b,D19c}. Below, we express it as the number of conical crossings, counted according to chirality. 
 
 Assume that $H \in \Ll$ and $\lambda_n(H)$ and $\lambda_{n+1}(H)$ degenerate conically at $\zeta_0$; and define $(J f)_j  = \lr{f, f_j}$, where $(f_1,f_2)$ is an orthogonal basis of $\ker\big(\lambda_n(\zeta_0) -H(\zeta_0)\big)$. We write a Taylor expansion of the $2 \times 2$ matrix $J  H(\zeta_0+\epsi)  J^*$ near $\epsi = 0$:
\begin{equation}\label{eq:1g}
 J  H(\zeta_0+\epsi)  J^* = J  H(\zeta_0)  J^* + \sum_{j=1}^3 (A_0 \epsi)_j \cdot  \sigma_j + O(\epsi^2),
 \end{equation}
where $\sigma_1, \sigma_2, \sigma_3$ are the standard Pauli matrices  and $A_0 \in M_3(\R)$.
Using the conical structure, $A_0$ is invertible -- see \eqref{eq:0r} below. The quantity $\sgn\big(\det(A_0) \big)$ is called the chirality of the conical point.

\begin{theorem}\label{thm:3} Let $H \in \Ll$, such that $\lambda_n(H)$ and $\lambda_{n+1}(H)$ degenerate conically precisely at $\zeta_1, \dots, \zeta_K$. If $\sgn\big(\det(A_1) \big), \dots, \sgn\big(\det(A_K) \big)$ are the associated chiralities, then
\begin{equation}\label{eq:3k}
c_1(\VVV_1) - c_1(\VVV_0) = \sum_{k=1}^K \sgn\big(\det(A_k) \big).
\end{equation}
\end{theorem}

Theorem \ref{thm:1} guarantees that $\Ll \neq \emptyset$ -- in fact, that $\Ll$ is a residual set.

 \subsection{Relation with adiabatic transport and bulk-edge correspondence} In this section, we explain the physical consequences of Theorems \ref{thm:1} and \ref{thm:3} on transport in adiabatic deformations of topological insulators.

Let $\HH_0$ and $\HH_1$ be two Hamiltonians satisfying \eqref{eq:1a}. Let $\{\HH_s\}_{s\in [0,1]}$ be a homotopy between $\HH_0$ and $\HH_1$;  extend $\HH_s$ by $\HH_0$ for $s \leq 0$ and by $\HH_1$ for $s \geq 1$. For $\delta > 0$, we define a Hamiltonian $\HH^{\delta}$ by
\begin{equation}\label{eq:3a}
\big(\HH^\delta \psi\big)_m = \big(\HH_{\delta m_2} \psi\big)_m,  \ \ \ \ \psi \in \ell^2(\Z^2,\C^N), \ \ m = (m_1,m_2) \in \Z^2.
\end{equation}
%\tr{adjust size of brackets}
The operator $\HH^\delta$ models a (spatial) deformation from $\HH_0$ to $\HH_1$ transversely to $\R e_1$, occurring at speed $\delta$. 

We are interested in the adiabatic scaling: $\delta \rightarrow 0$. This regime has an important place in the mathematical physics litterature; see e.g. \cite{S83,B84,PST02,%DGR06,
FT16}. It corresponds to changing $\HH_0$ to $\HH_1$ globally (i.e. on a scale $\delta^{-1} \gg 1$) while preserving translation-invariance locally (i.e. on a scale $\delta^{-1/2}$ -- note $1 \ll \delta^{-1/2} \ll \delta^{-1}$).

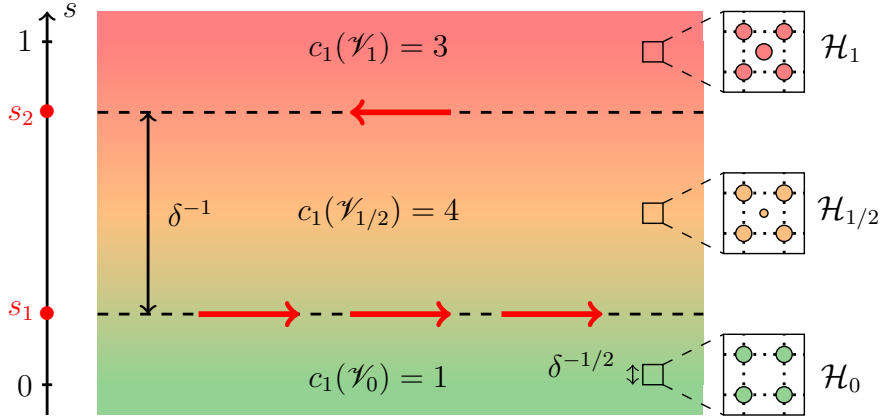
\begin{figure}
{\begin{tikzpicture}
  \node at (0,0) {\includegraphics[height=2.5in]{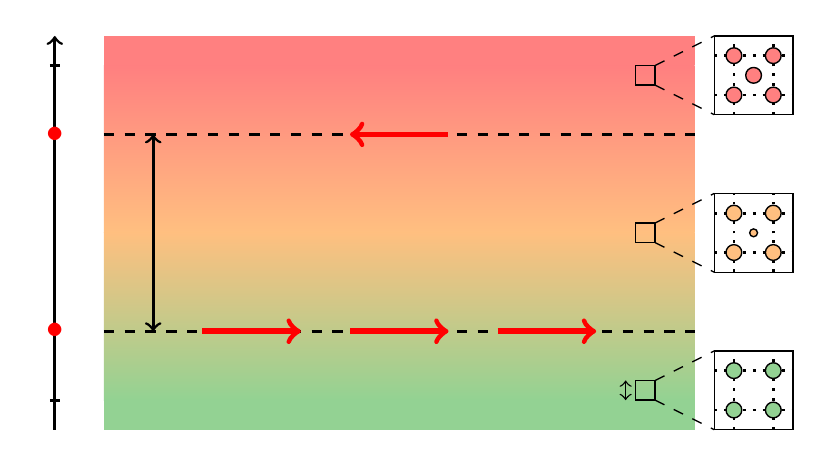}};

%\draw[dotted] (-5,-3) grid (5,3);

\node at (-4.6,2.7) {$s$};
\node at (-5.2,2.3) {$1$};
\node[red] at (-5.25,1.3) {$s_2$};
\node[red] at (-5.25,-1.3) {$s_1$};
\node at (-5.2,-2.3) {$0$};

\node at (-.5,-2.2) {$c_1(\VVV_0) = 1$};
\node at (-.5,0) {$c_1(\VVV_{1/2}) = 4$};
\node at (-.5,2.2) {$c_1(\VVV_1) = 3$};

\node at (-3,0) {$\delta^{-1}$};
\node at (2.2,-2) {$\delta^{-1/2}$};

\node at (5.65,-2.2) {$\HH_0$};
\node at (5.75
,0) {$\HH_{1/2}$};
\node at (5.65,2.2) {$\HH_1$};

   \end{tikzpicture}}
\captionsetup{width=0.5\textwidth}{
\caption{When deforming adiabatically two topological insulators $\HH_0$ and $\HH_1$, one must pass discontinuity channels for Chern numbers ($s= s_1, s_2$). These support a signed number of currents equal to the jump of Chern numbers. 
}\label{fig:6}}
\end{figure}

Generically, $\lambda_n(H_{\delta n_2})$ and $\lambda_{n+1}(H_{\delta n_2})$ do not degenerate for most values of $\delta n_2$. For such values, we can define the local Chern number of $\HH^\delta$ at $(n_1,\delta n_2)$: it is that of $H_{\delta n_2}$. The local Chern number is discontinuous at degeneracies, see Figure \ref{fig:6}.

In adiabatic domain-wall deformations of honeycomb structures, edge states arise and are concentrated near Dirac points (isotropic conical points) \cite{FLW16,LWZ17,D19a,DW19}. At leading order, they propagate according to an emerging Dirac operator, in the direction prescribed by chirality.

The analysis of \cite{FLW16,LWZ17,D19a,DW19} is local in nature and would extend beyond Dirac points.  Theorem \ref{thm:1} shows that degeneracies are generically conical. Hence, the Dirac-type propagation of edge states is universal in the adiabatic regime. See \S\ref{sec:1.5} and the appendix for more details.

In analogy with \cite{D19a,D19b,DW19}, the total number of edge states, signed according to propagation, is the sum over chiralities. From Theorem \ref{thm:2}, it is the total Chern number difference. Hence, \eqref{eq:3k} is a form of the bulk edge correspondence, the left-hand-side playing the role of an edge index -- see \cite{H93,KRS02,EGS05,ASV13,GP13,PS16,BKR17,D19c}. While the interface between $H_0$ and $H_1$ has width $\delta^{-1}$, the asymmetric transport described above concentrates in finitely many strips of width $\delta^{-1/2}$ (corresponding to jumps of local Chern number). This is a much thinner region. This concentration phenomenon -- valid only in the adiabatic regime -- is not captured by the bulk-edge correspondence.

\subsection{General statement}\label{sec:1.4} Theorem \ref{thm:1} will be the consequence of a stronger statement. Let $X$ be a smooth compact manifold of dimension 3.

\begin{defi}\label{def:2} If $H \in C^\infty(X,\EE)$, we say that $H$ has a degeneracy at $x_0 \in X$ if $H(x_0)$ admits repeated eigenvalues.\\
\indent We say that this degeneracy is conical if for some $n \in [0,N-1]$:
\begin{itemize}
\item[(i)] $\lambda_n\big( H(x_0) \big) = \lambda_{n+1}\big( H(x_0) \big)$ and all other eigenvalues of $H(x_0)$ are simple;
\item[(ii)] There exist $\ell \in C^\infty(\Omega,\R)$ and $q \in C^\infty\big(\Omega,[0,\infty)\big)$ with a non-degenerate critical value zero at $x_0$ such that
\begin{equation}\label{eq:1k}
\systeme{
\lambda_n\big(H(x)\big) & = & \ell(x) -  \sqrt{q(x)}
\\
\lambda_{n+1}\big(H(x)\big) & = & \ell(x) +  \sqrt{q(x)}}, \ \ \ \ x \text{ near } x_0.
\end{equation}
\end{itemize}
\end{defi}

For degeneracies of precisely double multiplicity, the mere estimate \eqref{eq:1f} is  equivalent to the smooth identity \eqref{eq:1k}; see \S\ref{sec:2.1}. In other words, Definition \ref{def:2} corresponds to Definition \ref{def:1}, with the additional requirement (i).

\begin{theorem}\label{thm:2} When $\dim(X) = 3$, the set
\begin{equation}\label{eq:1o}
\Mm = \big\{ H \in C^\infty(X,\EE) \ : \ \text{all degeneracies of M in  are conical } \big\}
\end{equation}
is dense and open in $C^\infty(X,\EE)$.
\end{theorem}

According to the von Neumann--Wigner theorem \cite{NW29}, $\EE \setminus \EE^*$ has codimension $3$ in $\EE$. Since $\dim(X) = 3$, the range $H(X)$ of $H$ has Hausdorff dimension at most $3$. Thus, generically, $H(X) \cap \left( \EE \setminus \EE^*\right)$ has Hausdorff dimension $0$; see Figure \ref{fig:3}. This result is closely related to various work about rarity of degenerate eigenvalues in mathematical physics; see e.g. \cite{C91,A95,T99}.

Theorem \ref{thm:2} completes \cite{NW29}: it shows that the degeneracies of a $3$-dimensional family of matrices are conical. In particular, generic elements in $C^\infty(X,\EE)$ have finitely many degeneracies.  As an immediate corollary with $X = \Tt^3$:

\begin{cor}\label{cor:1} The degeneracies of Bloch eigenvalues of a generic $\Z^3$-invariant Hamiltonian on $\ell^2(\Z^3,\C^N)$ are all conical.
\end{cor}

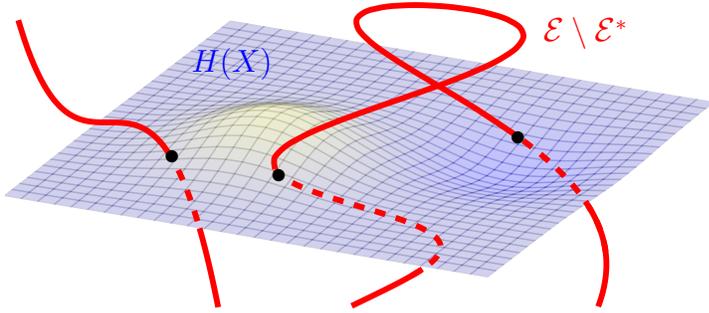
\begin{figure}
\floatbox[{\capbeside\thisfloatsetup{capbesideposition={right,center},capbesidewidth=3in}}]{figure}[\FBwidth]
{\caption{The range $H(X) \subset \EE$ of $H$ has (typical) dimension $3$, while $\EE \setminus \EE^*$ has codimension $3$. Generically, $H(X)$ and $\EE \setminus \EE^*$ intersect tranversely, along a set of dimension $0$.
}\label{fig:3}}
{\begin{tikzpicture}
  \node at (0,0) {\includegraphics[height=2in,trim={2cm 0cm 3cm 2cm}]{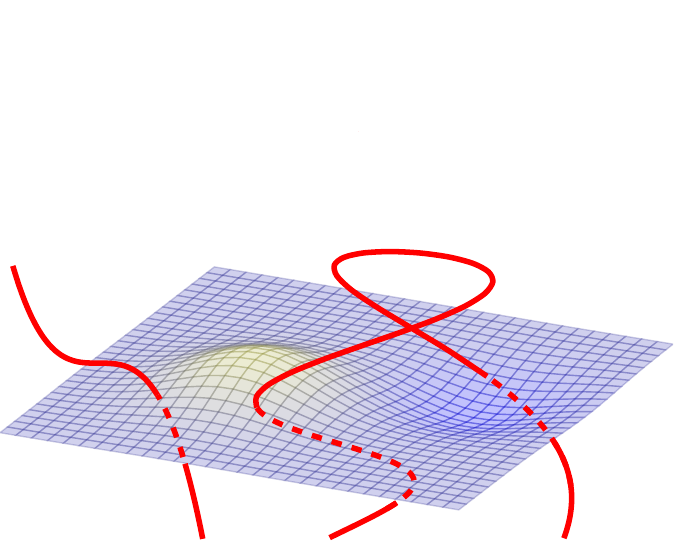}};

\node at (-0.38,-.5) {$\bullet$}; 
\node at (2.8,0) {$\bullet$}; 
 \node at (-1.8,-0.25) {$\bullet$}; 
 \node[white] at (4,-0.85) {$\bullet$}; 
 \node[white] at (-4,-0.2) {$\bullet$}; 
\node[blue] at (-1,1) {$H(X)$}; 
\node[red] at (3.7,1.4) {$\EE \setminus \EE^*$};

   \end{tikzpicture}}
\end{figure}

\subsection{Relation with existing work and perspectives} \label{sec:1.5}

The present work contrasts with earlier results in  tight-binding, quantum graphs, and continuous graphene models \cite{W47,C91,KP07,FW12,FLW18,L18}. These papers use the symmetries of the hexagonal lattice to show existence of Dirac points. 

The present paper is not symmetry-driven. It is instead topology-driven: conical points arise generically when trying to connect two topologically distinct Hamiltonian, and no other type of degeneracies may form. 

When connecting two topologically distinct Hamiltonians, assymetric currents appear along the interface: the celebrated edge states. Theorems \ref{thm:1} implies that \textit{generic} edge states of adiabatic systems on $\Z^2$ have amplitudes that, after rescaling, evolve according to a universal Dirac-like equation:
\begin{equation}\label{eq:0p}
\big( D_t - \Di(x_2, D_x) \big) \beta = 0, \ \ \ \ D_x = -i \p_x,
\end{equation}
where $\Di(x_2,\xi)$ is a family of $2 \times 2$ matrices which depends linearly in $x_2$ and $\xi$. We refer to the appendix for a formal derivation of \eqref{eq:0p}. A full proof would somewhat be transverse to this work; see \cite{FLW16,D19a,DW19} for a derivation in a slightly different context. See also \cite{FG03,F04,B19a,B19c} for direct work on \eqref{eq:0p}.

This Dirac-type propagation should also appear universally in continuous systems -- see e.g. \cite{RH08,FLW16,D19a,DW19} for honeycombs. This would require to extend Theorem \ref{thm:2} to differential operators. After some relatively standard reductions, the techniques developed here can treat systems on $L^2(\R^2)$ (corresponding to $N = \infty$). However they would yield a physically moot genericity result: it would hold within a class much larger than differential operators. We refer to \cite{C91,K16} for some interesting related conjectures, and formulate our own:

\begin{conj} The set
\begin{equation}
\big\{ V \in C^\infty(\R^3/\Z^3) :  \ \text{all degeneracies of Bloch eigenvalues of} -\Delta_{\R^3}+V \text{are conical }\big\}
\end{equation}
is dense and open in $C^\infty(\R^3/\Z^3)$.
\end{conj}

\subsection{Organization} We start with the proof of Theorem \ref{thm:2}. In \S\ref{sec:2}, we prove that $\Mm$ is open. This relies on the fact that conical points correspond precisely to critical values zero of the matrix discriminant. 
In \S\ref{sec:3}, we prove that $\Mm$ is dense. When $N=2$, this boils down to an algebraic identity combined with Sard's theorem. For $N \geq 3$, it relies on a reduction to the case $N=2$.

Theorem \ref{thm:1} follows from Theorem \ref{thm:2}, as explained in \S\ref{sec:4}. The proof of Theorem \ref{thm:3} is independent of the rest of the paper. It relies on arguments from \cite{D19b} -- see \S\ref{sec:5}. In the appendix, we explain the origin of the effective Dirac equation \eqref{eq:0p}.

\subsection{Notations}\label{sec:1.6}
\begin{itemize}
\item Given $N \in \N$, $\EE$ denotes the space of $N \times N$ Hermitian matrices , $\EE^* \subset \EE$ denotes matrices with simple eigenvalues; and $\FF \subset \EE$ consists of matrices with at most $N-2$ distinct eigenvalues. We provide these spaces with the (Hilbertien) norm $\|A\|^2 = \Tr_{\C^N}(A^2)$. 
\item Given a smooth compact manifold $X$, $\MM$ is the space $C^\infty(X,\EE)$; and $\Mm \subset \MM$ consists of elements in $\MM$ with only conical degeneracies -- see \S\ref{sec:1.4}. We fix a Riemannian structure on $X$, with Levi--Civita connection $\nabla$. The space $C^k(X,\EE)$ is the closure of $C^\infty(X,\EE)$ in $C^0(X,\EE)$, for the norm
\begin{equation}
\| H \|_{C^k} = \sup \left\{ \big\| H(x) \big\| + \big\| \nabla^k H(x) \big\|, \  x \in X \right\}, \ \ \ \ \ \ H \in \MM = C^\infty(X,\EE).
\end{equation}
It has a structure of Banach algebra. The space $\MM$ inherits a structure of complete metric space, with distance
\begin{equation}\label{eq:0b}
d(H,\tH) = \sum_{k=0}^\infty 2^{-k} \dfrac{\| H-\tH \|_{C^k}}{1+\| H-\tH \|_{C^k}}, \ \ \ \ \ \ H, \tH \in \MM.
\end{equation}
\item The space $C^\infty([0,1] \times \Tt^2,\EE)$ consists of Hermitian-valued smooth functions functions on $(0,1) \times \Tt^2$, whose derivatives extend continuously to $[0,1] \times \Tt^2$ -- also provided with the norm \eqref{eq:0b}. 
\item Given $H_0, H_1 \in C^\infty(\Tt^2, \EE)$ satisfying \eqref{eq:1e}, the space $\LL \subset C^\infty([0,1] \times \Tt^2,\EE)$ consists of smooth paths connecting $H_0$ to $H_1$. The space $\Ll \subset \LL$ consists of paths whose $n$-th and $n+1$-th eigenvalues degenerate conically -- see \S\ref{sec:1.1}.
\item  The Hausdorff dimension of a set $\SSS$ is denoted $\dim_\HH(\SSS)$.
\item The Pauli matrices are
\begin{equation}
\sigma_0 = \matrice{1 & 0 \\ 0 & 1}, \ \ \sigma_1 = \matrice{0 & 1 \\ 1 & 0}, \ \ \sigma_3 = \matrice{0 & -i \\ i & 0}, \ \ \sigma_3 = \matrice{1 & 0 \\ 0 & -1}.
\end{equation}
They form a basis of the space $\EE_0$ of traceless Hermitian $2 \times 2$ matrices.
\item If $x \in \R^3$ and $r > 0$,  $\Bb(x,r)$ is the ball centered at $x$ of radius $r$.
\end{itemize}

\noindent \textbf{Acknowledgments.} I thankfully acknowledge support from NSF
DMS-1440140 (MSRI, Fall 2019) and DMS-1800086, and from the Simons Foundation
through M. I. Weinstein’s Math+X investigator award \#376319.

\section{$\Mm$ is open}\label{sec:2}

We recall that $\MM = C^\infty(X,\EE)$. In this section, we show that the set $\Mm$ defined in \eqref{eq:1o} is open in $\MM$. In \S\ref{sec:2.1} we review the discriminant $D(A)$ of a matrix $A$. This is a quantity depending smoothly on the entries, whose zero set corresponds to matrices with degeneracies.

We then identify conical degeneracies of elements of $\MM$ with non-degenerate critical points of $D(H)$. Because of the stability of such  points,  $\Mm$ is open in $\MM$ -- see \S\ref{sec:2.2}.

\subsection{Discriminant and conical points}\label{sec:2.1} The discriminant of a matrix is the (square of the) Vandermonde determinant of the eigenvalues:
\begin{equation}\label{eq:1h}
D(A) = \prod_{j \neq k} \big(\lambda_j(A)-\lambda_k(A)\big) = \prod_{j < k} \big(\lambda_j(A)-\lambda_k(A)\big)^2, \ \ \ \ A\in \EE.
\end{equation}
It is a symmetric polynomial in $\lambda_1(A), \dots, \lambda_N(A)$. Thus, by the fundamental theorem of linear algebra, it is a polynomial in the quantities $\sum_{j=1}^m \lambda_j(A)^m = \Tr[A^m]$ -- see e.g. \cite[\S I.2]{M95}. In particular, $D(A)$ depends smoothly on $A$. 

The discriminant detects degenerate eigenvalues: $D(A) = 0$ if and only if $A \in \EE^*$. In fact, it even identifies conical degeneracies.

\begin{lem}\label{lem:1a} $H \in \MM$ has a conical degeneracy at $x_0$ if and only if $D \circ H$ has a non-degenerate critical value, zero, at $x_0$.
\end{lem}

\begin{rmk}\label{rem:2} No structure -- but that of a smooth manifold -- is required to define non-degenerate critical points of $u \in C^\infty(X,\R)$. A Riemannian structure on  $X$ allows us to consider the covariant Hessian $\nabla^2 u$; it is a symmetric endomorphism on $TX$ -- see e.g. \cite[\S2.1]{P06}. Non-degenerate critical points correspond to $du(x) = 0$ and $\nabla^2u(x)$ non-singular -- see e.g. \cite[\S5.12]{P06}.
\end{rmk}

\begin{proof} 1. We assume first that $H$ has a conical degeneracy at $x_0$. Let $\lambda_n\big( H(x_0) \big) = \lambda_{n+1}\big( H(x_0) \big)$ be the unique degenerate eigenvalue of $H(x_0)$. We write
\begin{equation}
D \circ H = \big( \lambda_{n+1}(H) - \lambda_n(H)\big)^2 \cdot F, \ \ \ \ F \de \prod_{\substack{j < k \\ (j,k) \neq (n,n+1)}} \big(\lambda_j(H)-\lambda_k(H)\big)^2.
\end{equation}
Using \eqref{eq:1k}, $D \circ H = q \cdot F$, where $q \in C^\infty(X,\R)$ has a non-degenerate critical value zero at $x_0$. From general theory, the eigenvalues of Hermitian matrices are Lipschitz in the entries -- see \cite[Proposition 6.2]{S10} -- hence a fortiori continuous. Thus $F$ is continuous. Moreover, since all eigenvalues of $H(x_0)$ are simple but $\lambda_n\big( H(x_0) \big) = \lambda_{n+1}\big( H(x_0) \big)$, $F(x_0) > 0$. We deduce that $D \circ H$ has a non-degenerate critical value zero at $x_0$.

2. Now we assume that $D \circ H$ has the non-degenerate critical value zero at $x_0$. Then there exists $\Omega$ neighborhood of $x_0$ such that 
\begin{equation}
x \in \Omega \setminus \{ x_0\} \ \ \Rightarrow \ \ D \circ H(x) \neq 0.
\end{equation}
In particular, for $x \in \Omega \setminus \{ x_0\}$, the eigenvalues $\lambda_j\big(H(x)\big)$ of $H(x)$ are simple -- hence smooth functions of $x$. 

3. Since $D \circ H(x_0) = 0$, $H(x_0)$ has at least one degenerate eigenvalue. Define
 \begin{equation}
 S = \Big\{ j \in [1,N-1] : \ \lambda_j\big( H(x_0) \big) = \lambda_{j+1}\big( H(x_0) \big) \Big\}.
 \end{equation}
Since eigenvalues of Hermitian matrices are Lipschitz functions of the entries, there exists $C > 0$ such that (after possibly shrinking $\Omega$):
\begin{equation}\label{eq:1j}
x\in \Omega, \ j \in S \ \ \Rightarrow \ \ 
\big|\lambda_j \big( H(x) \big) - \lambda_{j+1} \big( H(x) \big)\big| \leq C \big\| H(x)-H(x_0)\big\|.
\end{equation}
Let $J$ be the cardinal of $S$. From \eqref{eq:1h} and \eqref{eq:1j}, we deduce that for some $C' > 0$,
\begin{equation}
x\in \Omega \ \ \Rightarrow \ \  \big| D \circ H(x) \big| \leq C' \big\| H(x)-H(x_0)\big\|^{2J}.
\end{equation}
Since $H$ depends smoothly on $x$ and $D$ has a non-degenerate minimum at $x_0$, we deduce that $J \leq 1$. This implies that $H(x_0)$ has exactly $N-1$ distinct eigenvalues. Thus, if $n \in [1,N-1]$ is the unique integer such that $\lambda_n\big( H(x_0) \big) = \lambda_{n+1}\big( H(x_0)\big)$, then for $j \neq n, n+1$, $\lambda_j(H)$ are smooth in $\Omega$.

4. Let us fix a contour $\gamma \subset \C$ enclosing $\lambda_n\big(H(x_0)\big) = \lambda_{n+1}\big(H(x_0)\big)$ but no other eigenvalue of $H(x_0)$. After possibly shrinking $\Omega$, for $x \in \Omega$, $\gamma$ enclose $\lambda_n\big(H(x)\big)$ and $\lambda_{n+1}\big(H(x)\big)$ but no other eigenvalue of $H(x)$. Thus, 
\begin{equations}\label{eq:0a}
F_1(x) \de \Tr \left[ \int_\gamma z \big(z-H(x)\big)^{-1} \dfrac{dz}{2\pi i} \right] = \lambda_n\big(H(x)\big) + \lambda_{n+1}\big(H(x)\big) \ \ \  \text{and}
\\
F_2(x) \de \Tr \left[ \int_\gamma z^2 \big(z-H(x)\big)^{-1} \dfrac{dz}{2\pi i} \right] = \lambda_n\big(H(x)\big)^2 + \lambda_{n+1}\big(H(x)\big)^2 \ \ \ \ \ \ 
\end{equations}
are both smooth functions on $\Omega$. It follows that both 
\begin{equation}\label{eq:1m}
\ell \de \dfrac{\lambda_n(H) + \lambda_{n+1}(H)}{2} = \dfrac{F_1}{2} \ \ \ \text{and} \ \ \ q = \dfrac{\big(\lambda_{n+1}(H) - \lambda_n(H)\big)^2}{4} = \dfrac{2F_2-2F_1^2}{4}
\end{equation}
are smooth functions on $\Omega$.

5. The equation \eqref{eq:1m} imply that $\lambda_n(H) = \ell - \sqrt{q}$ and $\lambda_{n+1}(H) = \ell + \sqrt{q}$. Thus, it remains to show that $q$ has a non-degenerate critical point at $x_0$. Again, we write
\begin{equation}
D \circ H = q \cdot F, \ \ \ \ F \de \prod_{\substack{j < k \\ (j,k) \neq (n,n+1)}} \big(\lambda_j(H)-\lambda_k(H)\big)^2.
\end{equation}
We observe that $F$ is Lipschitz, with $F(0) \neq 0$. Hence, we have $D \circ H(x) = q(x) \big(1 + o(1)\big)$ near $x_0$; this implies
\begin{equation}
q(x) = D \circ H(x) \cdot \big( 1 + o(1) \big).
\end{equation}
Since $D \circ H(x)$ has a non-degenerate critical point at $x_0$, so does $q$. This completes the proof. \end{proof}

\subsection{$\Mm$ is open.}\label{sec:2.2} Here we prove that $\Mm$  -- defined in \eqref{eq:1o} -- is open in $\MM$. We fix a Riemannian structure on $X$ and consider Hessians of smooth functions on $X$ as symmetric endomorphisms of $TX$ -- see Remark \ref{rem:2}. Define $f : \MM \times X \rightarrow \R$ by
\begin{equation}\label{eq:0j}
f(A,x) \de \Det\big[ \big(\nabla^2 (D\circ A)\big)(x) \big]^2 + D \circ A(x).
\end{equation}

Fix $x \in X$ and $H \in \Mm$. If $H(x) \in \EE^*$, then $f(H,x) \geq D \circ H(x) > 0$. If $H(x) \notin \EE^*$, then $H(x)$ has a conical degeneracy at $x$. Because of Lemma \ref{lem:1a}, $D \circ H$ has a non-degenerate critical point at $x$, thus
\begin{equation}
f(H,x) \geq \Det\big[ \big(\nabla^2 (D\circ H)\big)(x) \big]^2 > 0.
\end{equation}
We deduce that $f(H,\cdot)$ is positive on $X$; since $X$ is compact, $\inf_{x \in X} f(H,x) > 0$.

Since $X$ is compact and $f(A,\cdot)$ depends only on the first two derivatives of $A$, there exists a constant $C$ depending only on $\| H \|_{C^2}$ such that
\begin{equation}\label{eq:0m}
\|B \|_{C^2} \leq 1 \ \ \Rightarrow \ \ 
\big|f(H+B,x) - f(H,x)\big| \leq C \|B\|_{C^2}.
\end{equation}
Since $\inf_{x \in X} f(H,x) > 0$, there exists $\epsi_0 > 0$ such that whenever $\| B \|_{C^2} \leq \epsi_0$, for every $x \in X$, $f(H+B,x) > 0$.

Hence, if $\| B \| \leq \epsi_0$ and $x \in X$, then either:
\begin{itemize}
\item $D \big( H(x) + B(x)\big) > 0$, that is $H(x) + B(x) \in \EE^*$;
\item or $D \big( H(x) + B(x)\big) > 0$ and $\Det\big[ \big(\nabla^2 D(H+B)\big)(x) \big]^2 >0$.
\end{itemize}
In the latter, $x$ is a non-degenerate critical point of $D(H+B)$. Thus $x$ is a conical degeneracy of $H+B$. This shows that $H+B \in \Mm$, hence $\Mm$ is open in $\MM$.

\section{$\Mm$ is dense}\label{sec:3}

In this section we show that $\Mm$ is dense in $\MM$. When $N = 2$, this follows from Sard's theorem and the fact that $D(A)$ is the sum of $3=\dim(X)$ squares depending smoothly on $A$; see \S\ref{sec:3.1}.

Two new problems arise for $N \geq 3$. Degeneracies can be more intricate: triple eigenvalues or pairs of double eigenvalues may arise. In \S\ref{sec:3.2}, we show that these are too rare to be significant in our problem. This will allow us to focus on $N \times N$ families of matrices with at least $N-1$ distinct eigenvalues.

The other obstacle is more serious: for $N \geq 3$, $D(A)$ is the sum of at least $5$ squares -- see \cite{D11}. Since $5 > \dim(X)$, the arguments of \S\ref{sec:3.1} do not naively extend. The key mechanism is that degeneracies of a $N \times N$ family $H \in \LL$ with at least $N-1$ distinct eigenvalues reduce \textit{locally} to those of a $2 \times 2$ family. This enables us to apply \textit{locally} the theory of \S\ref{sec:3.1}. The technical part in the proof of Theorem \ref{thm:2} consists of patching the local reductions -- see \S\ref{sec:3.3}-\ref{sec:3.4}.

\subsection{The case $N=2$}\label{sec:3.1} In this section only, we assume that $N=2$. This considerably simplifies that proof that $\Mm$ is dense -- and it will serve in the general situation. 

\begin{proof}[Proof that $\Mm$ is dense when $N=2$] When $N=2$, the Pauli matrices $\sigma_0 = \Id_2$, $\sigma_1$, $\sigma_2, \sigma_3$ form a basis of $\EE$. If $A = \sum_{j=0}^3 a_j \cdot \sigma_j$, then
\begin{equation}\label{eq:0w}
\sigma(A) = a_0 \pm |a|, \ \ \ \ D(A) = \| a \|^2, \ \ \ \text{where} \ \ \ a = [a_1,a_2,a_3]^\top.
\end{equation}

Let $H \in \LL$; we write $H(x) = \sum_{j=0}^3 h_j(x) \cdot \sigma_j$.  
Let $h = [h_1,h_2,h_3]^\top$ and 
\begin{equations}
\CC \de \big\{ t \in \R^3 : \ \exists x \in X, \ h(x) = t \text{ and } \rk \big(h'(x)\big) \leq 2 \big\} \\
=
\big\{ h(x) : \ x \in X, \ \rk \big(h'(x)\big) \leq 2 \big\}.
\end{equations}
According to Sard's theorem, the set $\R^3 \setminus \CC$ is dense in $\R^3$: given $\epsi > 0$, there exists $b \in \R^3 \setminus \CC$ with $\| b \| \leq \epsi$; see e.g.  \cite[\S1.7]{GP74}. Set $B = \sum_{j=1}^3 b_j \sigma_j$; we claim that all degeneracies of $H-B$ are conical. Indeed from \eqref{eq:0w}:
\begin{equation}
D\big(H(x) - B\big) = \big\| h(x) - b \big\|^2 = \sum_{j=1}^3 \big( h_j(x) -b_j \big)^2.
\end{equation}
From Lemma \ref{lem:1a}, $H-B$ can have a non-conical degeneracy at a point $x \in X$ only if $h(x) = b$ and $\rk \big(h'(x)\big) \leq 2$. This is always  excluded because $b \notin \CC$. Since $\epsi$ was arbitrary, we conclude that $\Mm$ is dense in $\MM$ when $N=2$. \end{proof}

\subsection{Removing high-multiplicity degeneracies}\label{sec:3.2} We go back to $N \neq 2$. In this section, we explain why we can focus our attention on family of matrices that always have at least $N-1$ distinct eigenvalues.

\begin{lem}\label{lem:1d} The set
\begin{equation}
\FF \de \big\{ A \in \EE : A \text{ has at most $N-2$ distinct eigenvalues} \big\}
\end{equation}
has Hausdorff dimension at most $N^2-6$.
\end{lem}

See \cite[\S2]{AS78} for related results -- but a different approach. 
Before giving the proof of Lemma \ref{lem:1d}, we discuss its consequences. We aim to prove that $\Mm$ is dense in $\MM$: given $H \in \MM$ and $\epsi > 0$, there exists $H_\epsi \in \Mm$ such that $d(H,H_\epsi) \leq 2\epsi$. Since $\dim X = 3$ and $\dim_\HH \FF \leq N^2-6$, the set 
\begin{equation}
\SSS = \big\{ H(x) - F : \ x \in X, \ F \in \FF\big\}
\end{equation}
has Hausdorff dimension at most $N^2-3$; thus $\EE \setminus \SSS$ has full measure. In particular, there exists $B \in \EE \setminus \SSS$ such that $\| B \| \leq \epsi$; and $H(x) + B \notin \FF$ for every $x \in X$. That is, $H+B \in C^\infty(X,\EE\setminus \FF)$.

Thus, to prove that $\Mm$ is dense in $\MM$, we just need to show that for every $H \in C^\infty(X,\EE\setminus\FF)$, there exists $H_\epsi \in \Mm$ with $d(H,H_\epsi) \leq \epsi$. %\tr{$\epsi$ again}

\begin{proof}[Proof of Lemma \ref{lem:1d}] 1. We observe that $\FF= \FF_1 \cup \FF_2$, where 
\begin{equation}
\FF_1 = \big\{A \in \EE : \ \text{$A$ has a triple eigenvalue} \big\}, \ \ \ \ 
\FF_2 = \FF \setminus \FF_1.
\end{equation}
Therefore, it suffices to show that $\FF_1$ and $\FF_2$ have Hausdorff dimension at most $N^2-6$.

2. We observe that $\FF_1 = \Phi(\GG_1,\R)$, where  $\GG_1$ consists of Hermitian $N \times N$ matrices of rank at most $N-3$; and $\Phi(B,\lambda) = B+\lambda$. We write
\begin{equation}\label{eq:0c}
\GG_1 = \bigcup_{j=0}^{N-3} \big\{ B \in \EE : \ \rk(B) = j\big\};
\end{equation}
and we recall that the sets in the RHS of \eqref{eq:0c} are smooth submanifolds of $\EE$, of dimension $N^2-(N-j)^2$ -- see e.g. \cite[\S1.4]{GP74}. Therefore, $\GG_1$ is a finite union of manifolds of dimensions up to $N^2-9$. We deduce that $\dim_{\HH}\GG_1 = N^2-9$ and $\dim_\HH \FF_1 = N^2-8$.

3. The set $\FF_2$ consists of matrices that have two distinct eigenvalues of multiplicity two but no triple eigenvalues.
We show that it has Hausdorff dimension at most $N^2-6$.  For $A_0 \in \FF_2$, there exists a unitary $N \times N$ matrix $U$ such that
\begin{equation}
U^* A_0  U=  \matrice{\lambda_1 \Id_2 & 0 & 0 \\ 0 & \Lambda & 0 \\ 0 & 0 & \lambda_2 \Id_2},
\end{equation}
where $\lambda_1 \neq \lambda_2$ and $\Lambda$ is a diagonal matrix of size $N-4$, with no diagonal coefficients equal to $\lambda_1$ or $\lambda_2$. In particular, both 
\begin{equation}
\matrice{\Lambda & 0 \\ 0 & \lambda_2 \Id_2} - \lambda_1 \ \ \text{and} \ \ \matrice{\lambda_1 \Id_2 & 0 \\ 0 & \Lambda } - \lambda_2
\end{equation}
are invertible $(N-2) \times (N-2)$ matrices. Therefore, there exists a neighborhood $\Omega \subset \EE$ of $A_0$ such that  for any $C \in \Omega$, we can write
\begin{equation}
 U^* A U = \matrice{C_1 & C_2 \\ C_2^* & C_3} = \matrice{D_3 & D_2 \\ D_2^* & D_1} ,
\end{equation}
where $C_3 - \lambda_1$ and $D_3-\lambda_2$ are $(N-2) \times (N-2)$ invertible matrices. 

4. If $R_1, R_2, R_3$ are consistently-sized matrices, with $R_1$ invertible,
\begin{equation}
\rk \left(\matrice{R_1 & R_2 \\ R_2^* & R_3} \right) = \rk(R_1)+ \rk \big(R_3-R_2^* R_1^{-1}R_2\big).
\end{equation}
This can be seen for instance from Schur's complement formula:
\begin{equation}\label{eq:0t}
\matrice{R_1 & R_2 \\ R_2^* & R_3} \matrice{\Id_{N-2} & -R_1^{-1}R_2 \\ 0 & \Id_2 } = \matrice{R_1 & 0 \\ R_2^* & R_3-R_2^* R_1^{-1}R_2}.
\end{equation}

Let $\Omega_1, \Omega_2 \subset \R$ be sufficiently small disjoint neighborhood of $\lambda_1, \lambda_2$ such that if $\mu_1 \in \Omega_1$ and $\mu_2 \in \Omega_2$,
\begin{equation}
\Phi_{\mu_1,\mu_2}(A) = \left( C_1 - \mu_1 - C_2^* (C_3-\mu_1)^{-1} C_2, \ D_1 - \mu_2 - D_2^* (D_3-\mu_2)^{-1} D_2 \right),
\end{equation}
from $\Omega$ to pairs of $2 \times 2$ Hermitian matrices, is well-defined. By \eqref{eq:0t}, $\Phi_{\mu_1,\mu_2}(A) = (0,0)$ if and only if $A-\mu_1$ and $A - \mu_2$ are of rank $N-2$; equivalently, if and only if $\mu_1$ and $\mu_2$ are two double eigenvalues of $A$. 

5. The map $\Phi_{\mu_1,\mu_2}$ is a local submersion at $A_0$. Indeed, we have
\begin{equation}
d\Phi_{\mu_1,\mu_2}(A_0) \cdot U \matrice{\epsilon_1 & 0 & 0 \\ 0 & 0 & 0 \\ 0 & 0 & \epsilon_2 }U^*  = (\epsilon_1,\epsilon_2).
\end{equation}
We note that $\Phi_{\mu_1,\mu_2}$ has range in  pairs of $2 \times 2$ Hermitian matrices, which has dimension $8$. Thus, by the local submersion theorem \cite[\S4]{GP74}, $\Phi_{\mu_1,\mu_2}^{-1}(0,0)$ is a submanifold of $\EE$ of dimension $N^2-8$. 

Using continuity of eienvalues, after potentially shrinking $\Omega$, we have
\begin{equation}
\FF_2 \cap \Omega = \bigcup_{(\mu_1,\mu_2) \in \Omega_1 \times \Omega_2} \Phi_{\mu_1,\mu_2}^{-1}(0,0).
\end{equation}
Since $\Omega_1 \times \Omega_2$ has dimension $2$, $\dim_\HH(\FF_2 \cap \Omega) \leq N^2-6$. Since $\Omega \subset \EE$ is a neighborhood of an arbitrary element $A_0 \in \FF_2$, $\FF_2$ is a countable union of sets of dimension at most $N^2-6$, thus it has dimension at most $N^2-6$.
\end{proof}

\subsection{Removing bad points: preparatory lemmas}\label{sec:3.3} Because of \S\ref{sec:3.2}, we focus (without loss of generalities) on $H \in C^\infty(X, \EE\setminus \FF)$: $H$ has, at all points of $X$, at least $N-1$ distinct eigenvalues. We  will show in \S\ref{sec:3.4} that $H$ is arbitrarily close to $\Mm$.

A naive generalization of \S\ref{sec:3.1} to $N \geq 3$ requires to write $D(A)$ as a sum of three squares depending smoothly on $A \in \EE$ -- see \eqref{eq:0w}. This is not possible: according to \cite{D11},` at least $5$ squares are necessary; see also \cite{I92,L98,P02,D11}. In \S\ref{sec:3.4}, we will get around by writing $D(H)$ \textit{locally} -- instead of \textit{globally} -- as a sum of $3$ squares. The present section lays out preparatory lemmas.

Fix $x_\star \in X$. According to the assumption, there exists $n_\star \in [1,N-1]$ such that
\begin{equation}
\lambda_1\big( H(x_\star) \big) < \dots < \lambda_{n_\star}\big( H(x_\star) \big) \leq \lambda_{n_\star+1}\big( H(x_\star) \big) < \dots < \lambda_N\big( H(x_\star) \big).
\end{equation}
Since eigenvalues are continuous functions of the entries, there exists an open neighborhood $X_\star \subset X$  of $x_\star$ such that
\begin{equation}\label{eq:2c}
x \in X_\star \ \ \Rightarrow \ \ 
\lambda_1\big( H(x) \big) < \dots < \lambda_{n_\star}\big( H(x) \big) \leq \lambda_{n_\star+1}\big( H(x) \big) < \dots < \lambda_N\big( H(x) \big).
\end{equation}
After potentially shrinking $X_\star$, there exists a ball $B(0,2r_\star) \subset \R^3$, and a smooth diffeomorphism $\phi_\star : B(0,2r_\star) \rightarrow X_\star$ with $\phi_\star(0) = x_\star$. We set $Y_\star = \phi\big(B(0,r_\star)\big) \subset X_\star$.

We observe that $x_\star \in Y_\star$. Thus, the collection of open sets $\{ Y_\star \}_{x_\star \in X}$ covers $X$ and we can pass to a finite collection, associated to points $x_1, \dots, x_P$. 

\begin{lem}\label{lem:1b} There exists $\delta_0 \in (0,1)$ such that for all $B \in \MM$ with $\| H-B\|_{C^0} \leq \delta_0$, for every $p \in [1,P]$,
\begin{equation}
x \in X_p \ \ \Rightarrow \ \ 
\lambda_1\big( B(x) \big) < \dots < \lambda_{n_p}\big( B(x) \big) \leq \lambda_{n_p+1}\big( B(x) \big) < \dots < \lambda_N\big( B(x) \big).
\end{equation}
\end{lem}

\begin{figure}
\floatbox[{\capbeside\thisfloatsetup{capbesideposition={right,top},capbesidewidth=3in}}]{figure}[\FBwidth]
{\caption{The proof that \vspace{1cm} $\Mm$ is dense goes as follows. \hspace{5cm}
(a) We first cover $X$ by topologically trivial open sets (here $Y_1, Y_2, Y_3$) on which the degeneracies of $H$ reduce  to those of a   $2 \times 2$ \vspace{1.4cm} system.  \hspace{5cm}
(b) In $Y_1$, the degenerate part of $H$ reduces to that of a $2 \times 2$ system. Via the procedure of \S\ref{sec:3.1}, we can produce $H_1$, arbitrarily close to $H$, with no bad points in $Y_1$. By Lemma \ref{lem:1c}, $\BB(H_1)$ is  a  small   perturbation of $\BB(H)\vspace{1.4cm} \setminus Y_1$. \hspace{5cm} 
(c) We repeat the procedure and produce $H_2$, arbitrarily close to $H_1$, with no bad points in $Y_2$. As bad points are stable, $\BB(H_2)$ is close to $\BB(H_1)$. In particular passing from $H_1$ to $H_2$ does not generate bad points back in  $Y_1 \setminus (Y_2 \cup Y_3)$, and \vspace{1.4cm} removes bad points in $Y_2$. \hspace{8cm} 
(d) We get new  systems $H_1, H_2, H_3$, recursively constructed, arbitrarily close to $H$, with no bad points in $Y_1, Y_1 \cup Y_2 \setminus Y_3, Y_1 \cup Y_2 \cup Y_3$, respectively. Since $Y_1 \cup Y_2 \cup Y_3$ cover $X$, $H_3$ is in $\Mm$ and is arbitrarily close to \vspace{1mm}  $H$. 
}\label{fig:4}}
{\begin{tikzpicture}
\definecolor{vert}{rgb}{.2, .5, .2}
\definecolor{purp}{rgb}{.7, .4, .07}

  \node at (-3,0) {\includegraphics[height=1.8in]{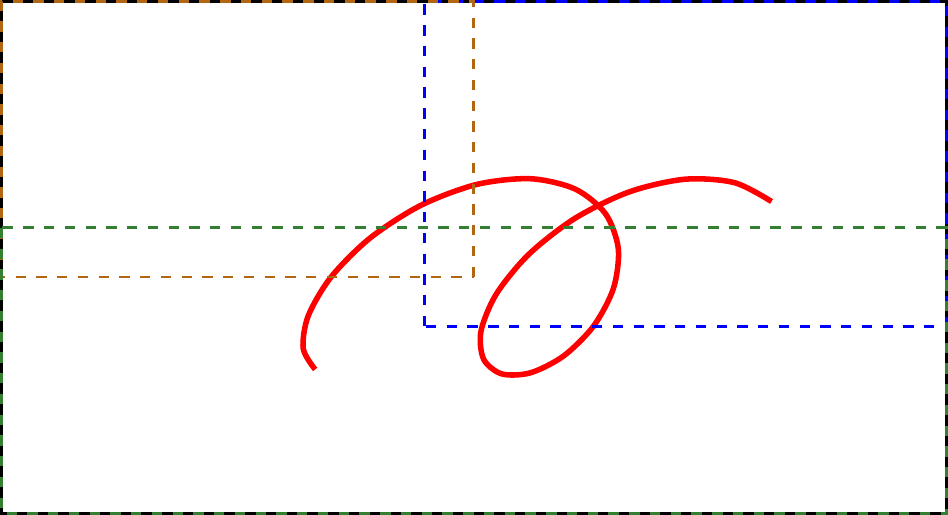}};
  \node at (-3,-5.5*1.071) {\includegraphics[height=1.8in]{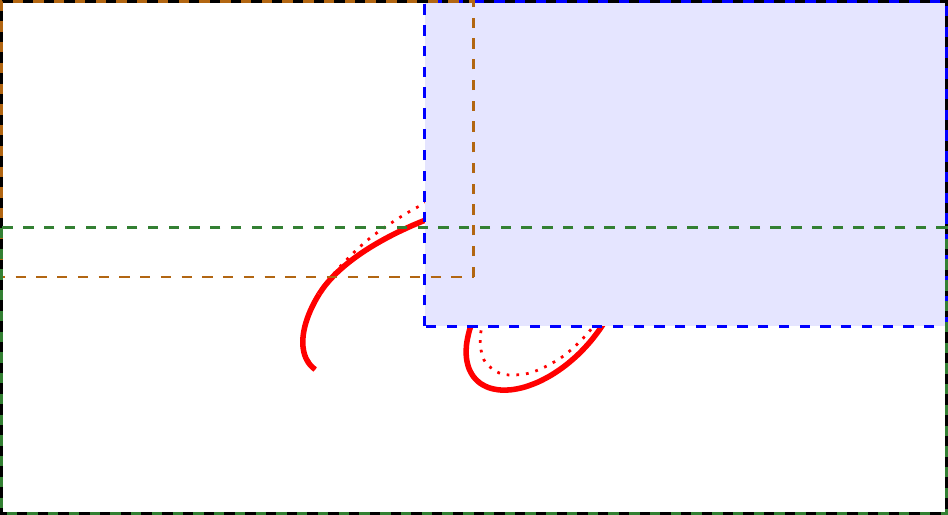}};
  \node at (-3,-11*1.071) {\includegraphics[height=1.8in]{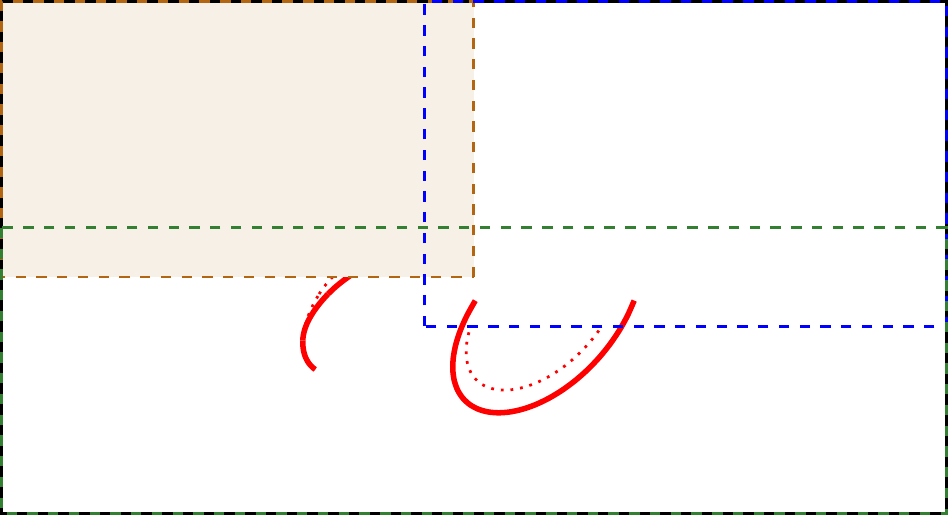}};
  \node at (-3,-16.5*1.071) {\includegraphics[height=1.8in]{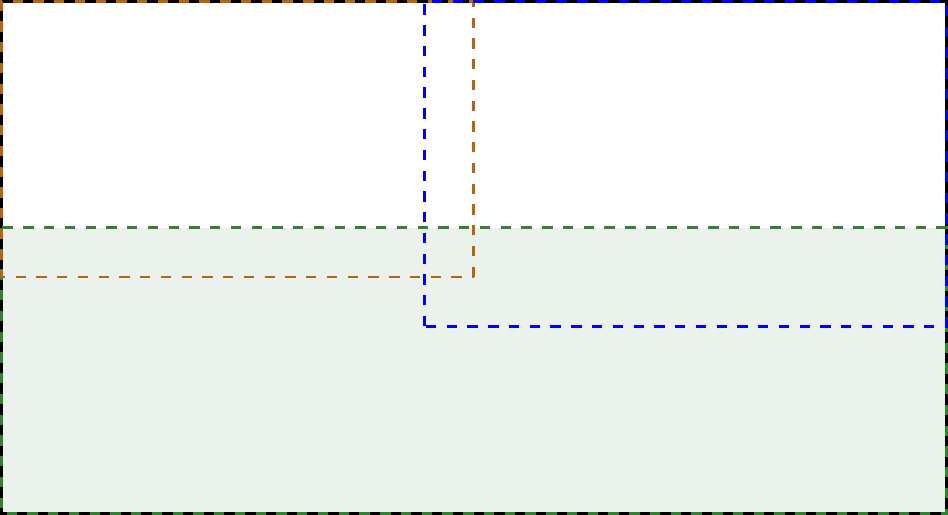}};
   
\node[blue] at (.8,1.9) {$Y_1$};   
\node[purp] at (-6.8,1.9) {$Y_2$};   
\node[vert] at (-6.8,-1.9) {$Y_3$};   
\node at (0.8,-1.9) {$X$};   
\node[red] at (-5.1,-1.1) {$\BB(H)$};

\node[blue] at (.8,1.9-5.5*1.071) {$Y_1$};   
\node[purp] at (-6.8,1.9-5.5*1.071) {$Y_2$};   
\node[vert] at (-6.8,-1.9-5.5*1.071) {$Y_3$};   
\node at (0.8,-1.9-5.5*1.071) {$X$};

\node[red] at (-5.15,-1.1-5.5*1.071) {$\BB(H_1)$};   

\node[blue] at (.8,1.9-11*1.071) {$Y_1$};   
\node[purp] at (-6.8,1.9-11*1.071) {$Y_2$};   
\node[vert] at (-6.8,-1.9-11*1.071) {$Y_3$};   
\node at (0.8,-1.9-11*1.071) {$X$};   
\node[red] at (-5.15,-1.1-11*1.071) {$\BB(H_2)$};

 \node[blue] at (.8,1.9-16.5*1.071) {$Y_1$};   
\node[purp] at (-6.8,1.9-16.5*1.071) {$Y_2$};   
\node[vert] at (-6.8,-1.9-16.5*1.071) {$Y_3$};   
\node at (0.8,-1.9-16.5*1.071) {$X$};     
  \node[red] at (-5.15,-1.1-16.5*1.071) {$\BB(H_3) = \emptyset$};

   \end{tikzpicture}}
\end{figure}

This result is a direct consequence of \eqref{eq:2c} with continuitiy of eigenvalues in the entries of the matrix -- \cite[Proposition 6.2]{S10}. 

Given $A \in \MM$, we say that $x \in X$ is a bad point of $A$ if $A$ has a non-conical degeneracy at $x$. We let $\BB(A)$ be the set of bad points of $A$; in particular, $A \in \Mm$ if and only if $\BB(A) = \emptyset$. Bad points are stable:

\begin{lem}\label{lem:1c} Let $A \in \MM$ and $Z \subset X$ be an open set such that $\BB(A) \subset Z$. Then there exists $\eta_0 \in (0,1)$ such that for all $B \in \MM$ with $\| B \|_{C^2} \leq \eta_0$, $\BB(A+B) \subset Z$.
\end{lem}

\begin{proof} Recall \eqref{eq:0j} and \eqref{eq:0m}: there exists $C > 0$ (depending on $\|A\|_{C^2}$) such that 
\begin{equations}\label{eq:1q}
\big\|B \big\|_{C^2} \leq 1 \ \ \Rightarrow \ \ 
\big|f(A+B,x) - f(A,x)\big| \leq C \|B\|_{C^2},
\ \ \ \ \ \text{where} 
\\
f(A,x) \de \Det\big[ \big(\nabla^2 (D\circ A)\big)(x) \big]^2 + D \circ A(x).
\end{equations}
Moreover, $f(A+B,x) = 0$ if and only if $x \in \BB(A+B)$. 

On the compact set $X \setminus Z$, $f(A,\cdot) > 0$. From \eqref{eq:1q}, if $\| B \|_{C^2}$ is sufficiently small, $f(A+B,\cdot) > 0$ on $X \setminus Z$. Thus $\BB(A+B) \subset Z$. This completes the proof.
\end{proof}

\subsection{Proof of Theorem \ref{thm:2}}\label{sec:3.4}  We refer to Figure \ref{fig:4} for a step-by-step pictorial explanation of the proof.

\begin{proof}[Proof that $\Mm$ is dense in $\MM$] 1. As explained in \S\ref{sec:3.2}, to prove density of $\Mm$ in $\MM$, it suffices to prove density of $\Mm$ in $C^\infty(X,\EE \setminus \FF)$. Let $H \in C^\infty(X,\EE \setminus \FF)$.
 Fix $0 < \epsi < \delta_0/4$, where $\delta_0$ is given by Lemma \ref{lem:1b}.  For each $p \in [0,P]$, we construct recursively $H_p \in \MM$ such that
\begin{equation}
d(H, H_p) \leq \big(1-2^{-p}\big) \epsi ; \ \ \text{ and  } \ \ \BB(H_p) \subset Z_p, \ \ \ Z_p \de Y_{p+1} \cup \dots \cup Y_P.
\end{equation}
In particular, $H_P$ will satisfy $d(H,H_P) \leq \epsi$ and $\BB(H_P)  = \emptyset$. 

For $p=0$, we simply take $H_0 = H$. For $p \geq 1$, we proceed by induction: we assume that $H_{p-1}$ is constructed and we want to construct $H_p$. 

2. For $x \in X_p$, let $\VV(x)$ be the eigenspace of $H_{p-1}(x)$ associated to the eigenvalues $\lambda_{n_p}\big( H_{p-1}(x) \big)$ and  $\lambda_{n_p+1}\big( H_{p-1}(x) \big)$. Since $d(H, H_{p-1}) \leq \epsi$, $\| H-H_{p-1}\|_{C^0} \leq \delta_0$. Thus Lemma \ref{lem:1b} implies that for every $x \in X_p$,
\begin{equation}\label{eq:1v}
\lambda_1\big( H_{p-1}(x) \big) < \dots < \lambda_{n_p}\big( H_{p-1}(x) \big) \leq \lambda_{n_p+1}\big( H_{p-1}(x) \big) < \dots < \lambda_N\big( H_{p-1}(x) \big).
\end{equation}
Because of \eqref{eq:1v}, $\VV(x)$ induces a rank-two vector bundle over $X_p$; and so does $\VV(x)^\perp$. Since $X_p$ is diffeomorphic to a ball in $\R^3$, $\VV$ and $\VV^\perp$ are trivial vector bundles -- see e.g. \cite[\S1.3]{M01}. Therefore, they both admit unitary frames. This means that there exists $U \in C^\infty\big(X_p,U(N)\big)$ such that for all $x \in X_p$,
\begin{equation}\label{eq:1u}
H_{p-1}(x)  = U(x) \matrice{J(x) & 0 \\ 0 & J(x)^\perp} U(x)^*, \ \ \ \ \text{where:}
\end{equation}
\begin{itemize}
\item $J(x)$ is a $2 \times 2$ Hermitian matrix depending smoothly on $x \in X_p$, with eigenvalues $\lambda_{n_p}\big( H_{p-1}(x) \big)$ and  $\lambda_{n_p+1}\big( H_{p-1}(x) \big)$;
\item $J(x)^\perp$ is a $(N-2) \times (N-2)$ Hermitian matrix depending smoothly on $x \in X_p$, with simple eigenvalues $\lambda_j\big( H_{p-1}(x) \big), \ j \notin\{ n_p, n_p+1\}$. 
\end{itemize}

3. Let $\chi \in C^\infty(X,\R)$ be equal to $1$ on a neighborhood of $Y_p$, with support contained in $X_p$. Let $B$ be a Hermitian $2 \times 2$ matrix and define
\begin{equation}\label{eq:1r}
H_p(x)  \de  H_{p-1}(x) + \chi(x)^2 \cdot U(x) \matrice{ B & 0 \\ 0 & 0}U(x)^*.
\end{equation}
We note that
$H_p \in \MM$: $\chi = 0$ when $U$ is not well-defined. As $C^k(X,\R)$ is an algebra, 
\begin{equation}
\| H_p-H_{p-1} \|_{C^k} \leq \az_k \|B\|, \ \ \ \ \az_k \de C_k \| \chi U \|_{C^k} \| \chi U^* \|_{C^k}.
\end{equation}
Using that $s \mapsto s(1+s)^{-1}$ increases on $[0,\infty)$,
\begin{equation}\label{eq:2d}
d(H_p,H_{p-1}) = \sum_{k=0}^\infty 2^{-k} \dfrac{\| H_p-H_{p-1} \|_{C^k} }{1+\| H_p - H_{p-1}\|_{C^k}} \leq \sum_{k=0}^\infty 2^{-k} \dfrac{\az_k \|B\| }{1+\az_k \|B\|}.
\end{equation}
We split the sum in the RHS two parts, depending whether $\az_k$ is  larger than $\|B\|^{-1/2}$. Since $s(1+s)^{-1} \leq \min(1,s)$, we deduce that
\begin{equation}
\sum_{\az_k \leq \|B\|^{-1/2} } 2^{-k} \dfrac{\az_k \|B\| }{1+\az_k \|B\|} \leq 2\|B\|^{1/2}, \ \ \ \ \sum_{\az_k > \|B\|^{-1/2} } 2^{-k} \dfrac{\az_k \|B\| }{1+\az_k \|B\|} \leq 2^{-k_B+1},
\end{equation}
where $k_B$ is the smallest integer such that $\az_k > \| B \|^{-1/2}$ (with $k_B = \infty$ no such integer exist). In particular, $k_B \rightarrow \infty$ as $\| B \| \rightarrow 0$. Going back to \eqref{eq:2d}, we deduce that
\begin{equation}\label{eq:2e}
d(H_p,H_{p-1}) \leq 2 \left(\| B \|^{1/2} + 2^{-k_B}\right) \rightarrow 0 \ \ \ \text{ as } \ \ \ \| B \| \rightarrow 0.
\end{equation}

4. Let $\eta_0$ associated to $H_{p-1}$ and $Z_{p-1}$  by Lemma \ref{lem:1c}. 
Thanks to \S\ref{sec:3.1} and \eqref{eq:2e} can find a Hermitian $2 \times 2$ matrix $B$ with the two following conditions:
\begin{itemize}
\item All degeneracies of $J(x) + B$ in $X_p$ are conical;
\item $d(H_p,H_{p-1}) \leq \min\left(2^{-p} \epsi,\eta_0/8\right)$.
\end{itemize}

The recursion assumption $d(H,H_{p-1}) \leq (1-2^{-p-1}) \epsi$ and $d(H_p,H_{p-1}) \leq 2^{-p} \epsi$ yield  $d(H, H_p) \leq \big(1-2^{-p}\big) \epsi$. Moreover, $d(H_p,H_{p-1}) \leq \eta_0/8$ implies $\| H_p - H_{p-1} \|_{C^2} \leq \eta_0$. From Lemma \ref{lem:1c} and the recursion assumption $\BB(H_{p-1}) \subset Z_{p-1}$, $\BB(H_p) \subset Z_{p-1}$. 

4. To complete the recursion, it remains to show that $\BB(H_p)  \subset Z_p$; equivalently, that $H_p$ has no bad degeneracies in $Y_p$. When $\chi(x) = 1$ (i.e. on a neighborhood of $Y_p$), 
\begin{equation}\label{eq:1t}
H_p(x) = U(x) \matrice{J(x) + B & 0 \\ 0 & J(x)^\perp}U(x)^*.
\end{equation}
Using \eqref{eq:1u}, the identity \eqref{eq:1t} implies that when $\chi(x)=1$, the eigenvalues of $H_p(x)$ are: $\lambda_j\big( H_{p-1}(x) \big)$ for $j \neq n_p, n_p+1$; and $\lambda_j\big( J(x) +  B \big)$, $j =1,2$. 

From \eqref{eq:1v}, the only possible degeneracies of $H_p$ in $\{ \chi = 1 \}$ arise from  $\lambda_1\big( J + B \big)$ and $\lambda_2\big( J +  B \big)$. By definition of $B$, all such degeneracies are conical. Since $Y_p \subset \{ \chi=1\}$, we get $\BB(H_p) \cap Y_p = \emptyset$. This completes the recursion and the proof of Theorem \ref{thm:2}.
\end{proof}

\section{Proof of Theorem \ref{thm:1}}\label{sec:4}

\begin{proof}[Proof that $\Ll$ is open in $\LL$] The proof is similar to \S\ref{sec:2}. Fix $H \in \Ll$; let $\{ \zeta_1, \dots, \zeta_J\}$ be the (finite) set of points of $[0,1] \times \Tt^2$, such that $\lambda_n(H)$ and $\lambda_{n+1}(H)$ degenerate. 

For each $j \in [1,J]$, let $\gamma_j$ be a contour enclosing $\lambda_n\big(H(\zeta_j)\big) = \lambda_{n+1}\big(H(\zeta_j)\big)$, but no other eigenvalue of $H(\zeta_j)$. Using continuity of eigenvalues, there exist $\epsi_0$ and $r_0 > 0$ such that for $B \in \MM$ with $\| B \|_{C^2} \leq \epsi_0$ and $\zeta \in \Bb(\zeta_j,r_0)$,  $\gamma_j$ encloses
$\lambda_n\big(H(\zeta) + B(\zeta)\big)$ and  $\lambda_{n+1}\big(H(\zeta) + B(\zeta)\big)$ but no other eigenvalues of $H(\zeta) + B(\zeta)$.

Without loss of generality, the balls $\Bb(\zeta_j,r_0)$ are disjoints. For $\zeta \in \Bb(\zeta_j,r_0)$, introduce, similarly to \eqref{eq:0a},
\begin{equations}
G_1(\zeta,B) \de \Tr \left[ \int_{\gamma_j} z \big(z-H(\zeta)-B(\zeta)\big)^{-1} \dfrac{dz}{2\pi i} \right] = \sum_{j = n}^{n+1}\lambda_j\big(H(\zeta)+B(\zeta)\big),
\\
G_2(\zeta,B) \de \Tr \left[ \int_{\gamma_j} z^2 \big(z-H(\zeta)-B(\zeta)\big)^{-1} \dfrac{dz}{2\pi i} \right] = \sum_{j = n}^{n+1} \lambda_j\big(H(\zeta)+B(\zeta)\big)^2,
\\
G(\zeta,B) \de 2G_2(\zeta,B) - G_1(\zeta,B)^2 = \Big(\lambda_{n+1}\big(H(\zeta)+B(\zeta)\big) - \lambda_n\big(H(\zeta)+B(\zeta)\big)\Big)^2.
\end{equations}
We note that $G(\zeta_j,0) = 0$ hence $\nabla^2_\zeta G(\zeta_j,0) > 0$, because $\lambda_n(H)$ and $\lambda_{n+1}(H)$ may only degenerate conically. The identity $G = 2G_2-G_1^2$ and the Cauchy representation of $G_1$ and $G_2$ imply that for some $C > 0$ and all $\zeta \in \Omega$,
\begin{equation}
\big| \nabla^2_\zeta G(\zeta,B) - \nabla^2_\zeta G(\zeta,0) \big| \leq C \| B \|_{C^2}.
\end{equation}
Therefore, after possibly shrinking $\epsi_0$ and $\Omega$, 
\begin{equation}
\| B \|_{C^2} \leq \epsi_0, \ \zeta \in \Omega \ \ \Rightarrow \ \ \nabla^2_\zeta G(\zeta,B) > 0.
\end{equation} 
Thus, if $\lambda_n\big(H(\zeta) + B(\zeta)\big) = \lambda_{n+1}\big(H(\zeta) + B(\zeta)\big)$ for $\zeta\in \Omega$, then this degeneracy is conical. Finally, after shrinking $\epsi_0$, $\lambda_n\big(H + B\big)$ and $\lambda_{n+1}\big(H + B\big)$ cannot degenerate outside $\Omega$. This shows that $H+B \in \Ll$: $\Ll$ is open in $\LL$.
\end{proof}

\begin{proof}[Proof that $\Ll$ is dense in $\LL$] 1. We show that $\Ll$ is dense in $\LL$. 
Since eigenvalues are Lipschitz functions of the matrix entries, we deduce from \eqref{eq:1e} that there exists $\eta_0 \in (0,1)$ such that for every $T \in \EE$,
\begin{equation}\label{eq:2a}
\| T \| \leq 4\eta_0, \ \ \xi \in \Tt^2 \ \ \Rightarrow  \ \ \systeme{
\lambda_{n+1}\big(H_0(\xi)+T\big) > \lambda_n\big(H_0(\xi)+T\big) \\ \lambda_{n+1}\big(H_1(\xi)+T\big) > \lambda_n\big(H_1(\xi)+T\big)}.
\end{equation}

2. Let $H \in \LL$: $H$ is smooth on $(0,1) \times \Tt^2$, with bounded derivatives; and connects $H_0$ to $H_1$. Seeley's operator \cite{S65} extends $H$ as an element of $C_0^\infty((-\pi,\pi) \times \Tt^2,\EE)$, thus as an element of $C^\infty(\Tt^3,\EE)$ (still denoted $H$).

Let $\chi_0$, $\chi_1 \in C_0^\infty(\Tt^1,[0,1])$ with $\chi_0(0) = \chi_1(1) = 1$ and
\begin{equation}
\supp(\chi_0) \subset (-\delta_0,\delta_0)/(2\pi \Z), \ \ \supp(\chi_0) \subset (1-r_0,1+r_0)/(2\pi \Z), \ \ \ \ r_0 \de \dfrac{\eta_0}{1+\| H \|_{C^1}}.
\end{equation}
For $\HH \in C^\infty(\Tt^3,\EE)$, we introduce
\begin{equation}\label{eq:1b}
\Hh(s,\xi) = \HH(s,\xi) + \chi_0(s) \big( H_0(\xi) - \HH(0,\xi) \big) + \chi_1(s) \big( H_1(\xi) - \HH(1,\xi) \big). 
\end{equation}
We observe that $\Hh$ restricts to $[0,1] \times \Tt^2$ as an element of $\LL$: it varies smoothly with $(s,\xi)$ and connects $H_0$ to $H_1$.

3. Fix $\epsi > 0$. Using \eqref{eq:1b} and that $C^k(\Tt^3,\EE)$ is an algebra, we have $\| H -\Hh \|_{C^k} \leq C_k \| H -\HH \|_{C^k}$ for some $C_k > 0$. As in Step 3 in \S\ref{sec:3.4} there exists $\eta_1 \in (0,\eta_0)$ with
\begin{equation}
d(H,\HH) \leq \eta_1 \ \ \Rightarrow \ \ d(H,\Hh) \leq \epsi.
\end{equation}
We now demand that $\HH \in \Mm$, and $d(H,\HH) \leq \eta_1$; such $\HH$ exist by Theorem \ref{thm:2}. Under these conditions, $\Hh$ defined by \eqref{eq:1b} satisfies $d(H,\Hh) \leq \epsi$; we claim that $\lambda_n(\Hh)$ and $\lambda_{n+1}(\Hh)$ can only degenerate conically in $[0,1] \times \Tt^2$.

4. For $(s,\xi) \in (r_0,1-r_0) \times \Tt^2$, we have $\Hh(s,\xi) = \HH(s,\xi)$. Since $\HH \in \Mm$, we deduce that $\lambda_n(\Hh)$ and $\lambda_{n+1}(\Hh)$ can only degenerate conically in $(r_0,1-r_0) \times \Tt^2$.

For $(s,\xi) \in [0,r_0] \times \Tt^2$, we have
\begin{equations}
\big\| \Hh(s,\xi) - H_0(\xi) \big\| \leq \| \HH(s,\xi)-H_0 \| + \| H_0(\xi) -\HH(0,\xi) \| 
\\
\leq  \| \HH(s,\xi)-\HH(0,\xi) \| + 2\| H_0(\xi) -\HH(0,\xi) \|
\\
\leq r_0 \| \HH \|_{C^1} + 2 \delta \leq r_0 (\delta + \|H\|_{C^1}) + 2 \delta \leq 3\eta_0.
\end{equations}
In the last line, we used the definition of $r_0$ and the inequality $\delta < \eta_0 < 1$. Thanks to \eqref{eq:2a}, we deduce that $\lambda_n(\Hh)$ and $\lambda_{n+1}(\Hh)$ cannot cross in $[0,r_0] \times \Tt^2$. A similar argument shows that they cannot cross in $[1-r_0,1] \times \Tt^2$. 

Hence, the restriction of $\Hh$ to $[0,1] \times \Tt^2$ is in $\Ll$; and $d(H,\Hh) \leq \epsi$. Since $\epsi$ was arbitrary, $\Ll$ is dense in $\LL$. This completes the proof of Theorem \ref{thm:1}.
 \end{proof}

\section{Chern number difference}\label{sec:5}

\begin{proof}[Proof of Theorem \ref{thm:3}] 1. We start with a few notations and definitions. Let $H \in \Ll$. Let $\RRR$ be the set of points $\zeta = (s,\xi) \in [0,1] \times \Tt^2$ such that $\lambda_n\big( H(\zeta) \big) < \lambda_{n+1}\big(H(\zeta)\big)$.  For $\zeta \in \RRR$, we can represent the projector $\Pi_n(\zeta)$ to the first $n$ eigenspaces of $H(\zeta)$ as a Cauchy integral:
\begin{equation}\label{eq:0f}
\Pi_n(\zeta) = 
\dfrac{1}{2\pi i}\oint_{\gamma_n(\zeta)} \big( z - H(\zeta) \big)^{-1} dz,
\end{equation}
where $\gamma_n(\zeta) \subset \C$ encloses $\lambda_1\big( H(\zeta) \big), \dots, \lambda_n\big( H(\zeta) \big)$ but no other eigenvalue of $H(\zeta)$. If $\{s \} \times \Tt^2 \subset \RRR$, then $\Pi_n(s,\cdot)$ induces a vector bundle over $\Tt^2$: the fiber at $\xi \in \Tt^2$ is $\Range\big( \Pi_n(s,\xi)\big)$. We let $\SSS$ be the set of $s \in [0,1]$ such that $\VVV_s$ is not well-defined -- equivalently, $\SSS = \{s \in [0,1]: \ \exists \xi \in \Tt^2, (s,\xi) \notin \RRR\}$.

For $\zeta \in \RRR$, we define
\begin{equation}\label{eq:0g}
B_n(\zeta) =  \Tr_{\C^N} \left( \Pi_n(\zeta) \big[ \p_{\xi_1} \Pi_n(\zeta),  \p_{\xi_2} \Pi_n(\zeta) \right).
\end{equation}
This is a smoothly varying function on $\RRR$, that interprets as the Berry curvature. In particular, $B_n(\xi)d\xi$ is a two-form; and $B_n(\xi)$ is additive: if $\zeta \in \RRR$ and $\lambda_{n-1}\big( H(\zeta) \big) < \lambda_n\big(H(\zeta)\big)$, then $B_n(\zeta) = B_{n-1}(\zeta) + b(\zeta)$, where:
\begin{itemize}
\item $B_{n-1}(\zeta)$ is associated with the projector $\Pi_{n-1}(\zeta)$ to the first $n-1$ eigenspaces of $H(\zeta)$ -- see \eqref{eq:0f}, \eqref{eq:0g} with $n$ replaced by $n-1$;
\item $b(\zeta)$ is associated to the rank-one projector $\pi(\zeta)$ to $\ker\big( \lambda_n \big(H(\zeta)\big) - H(\zeta) \big)$:
\begin{equation}\label{eq:0y}
b(\zeta) =  \Tr_{\C^N} \left( \pi(\zeta) \big[ \p_{\xi_1} \pi(\zeta),  \p_{\xi_2} \pi(\zeta) \big]\right).
\end{equation}
\end{itemize}
For $s \in [0,1] \setminus \SSS$, the Chern number of $\VVV_s$ is the integer
\begin{equation}
c_1(\VVV_s) = \dfrac{1}{2\pi i} \int_{\Tt^2} B_n(s,\xi) d\xi.
\end{equation}

In Step 6, we will use the space of $2 \times 2$ traceless Hermitian matrices $\EE_0$. This space is equipped with the Hermitian inner product $\lr{T_1,T_2} = \Tr(T_1T_2)$; the Pauli matrices $\sigma_1, \sigma_2, \sigma_3$ form an orthonormal basis. If $\tsigma_1, \tsigma_2, \tsigma_3$ is another orthonormal basis, then there exists $U \in SU(2)$ (unique up to multiplication by $\pm \Id_2$) and $\epsilon \in \{\pm 1\}$ such that
\begin{equation}\label{eq:0h}
\tsigma_k  = \epsilon \cdot U \sigma_k U^*, \ \ \ \ 1 \leq k \leq 3.
\end{equation}
This is precisely the content of the isomorphism between $SU(2) / \{\pm \Id_2\}$ and $SO(3)$; see e.g. \cite[\S4.2]{S05}. The number $\epsilon \in \{\pm 1\}$ reads as the determinant of the (orthogonal) matrix of the basis $(\tsigma_1, \tsigma_2, \tsigma_3)$ in the basis $(\sigma_1, \sigma_2, \sigma_3)$.

2. Since $H \in \Ll$, the sets $[0,1] \setminus \RRR$ and $\SSS$ are finite. The map $s \mapsto c_1(\VVV_s)$ is well--defined on $[0,1] \setminus \SSS$. Since it is integer-valued, it is locally constant on each sub-interval of $[0,1] \setminus \SSS$. We deduce that
\begin{equations}\label{eq:0d}
c_1(\VVV_1) - c_1(\VVV_0) = \lim_{\delta \rightarrow 0^+} \ \sum_{s_\star \in \SSS} c_1(\VVV_{s_\star+\delta}) - c_1(\VVV_{s_\star-\delta})
\\
 = \dfrac{1}{2\pi i} \sum_{s_\star \in \SSS} \lim_{\delta \rightarrow 0^+} \int_{\Tt^2} \big(B_n(s_\star+\delta,\xi) - B_n(s_\star-\delta,\xi) \big) d\xi
\end{equations}
It remains to compute each individual summand in the RHS of \eqref{eq:0d}. For that, we use the techniques developed in \cite[\S2]{D19b} -- and we refer to that paper for full details.

3. Fix $s_\star \in \SSS$; let $\ZZZ$ be the set of points $\xi \in \Tt^2$ such that $(s_\star,\xi) \notin \RRR$. Using that $B(\zeta)$ depends smoothly on $\zeta \in \RRR$, we deduce that for $r$ sufficiently small,
\begin{equation}\label{eq:0e}
\int_{\Tt^2} \big(B_n(s_\star+\delta,\xi) - B_n(s_\star-\delta,\xi) \big) d\xi = \sum_{\xi_\star \in \ZZZ} \int_{|\xi-\xi_\star| \leq r} \big(B_n(s_\star+\delta,\xi) - B_n(s_\star-\delta,\xi) \big) d\xi   + O(\delta).
\end{equation}
We refer to the proof of \cite[Lemma 2.1]{D19b} for details. Hence, it suffices to estimate each summand in the RHS of \eqref{eq:0e}.

4. Fix $\zeta_\star = (s_\star,\xi_\star) \in \ZZZ$. Since $H \in \Ll$, $\lambda_n(H)$ and $\lambda_{n+1}(H)$ degenerate conically at $\zeta_\star$. In particular, $\lambda_n\big(H(\zeta_\star)\big) > \lambda_{n-1}\big(H(\zeta_\star)\big)$. Therefore, $\Pi_{n-1}(\zeta)$ -- hence $B(\zeta)$ -- depend smoothly on $\zeta$ near $\zeta_\star$.
Using the additivity of the Berry curvature, we get
\begin{equation}\label{eq:0l}
B_n(\zeta_\star+\epsi) = b(\zeta_\star+\epsi) + O(1),
\end{equation}
for $\epsi$ sufficiently small. We refer to the proof of \cite[(2.21)]{D19b} for details. It remains to understand $b(\zeta)$ near $\zeta_\star$, hence $\pi(\zeta)$ and its derivatives near $\zeta_\star$.

5. Let $\{f_1,f_2\}$ be an orthonormal basis of $\ker\big( H(\zeta_\star)-\lambda_n\big( H(\zeta_\star) \big)$. We define
\begin{equation}
J : \C^N \rightarrow \C^2, \ \ \ Jf = \matrice{\lr{f,f_1} \\ \lr{f,f_2}}.
\end{equation}
We write a Taylor development of the $2 \times 2$ matrix $J H(\zeta) J^*$ near $\zeta_\star$:
\begin{equations}
J H(\zeta_\star+\epsi) J^* = J H(\zeta_\star) J^* + \sum_{j=1}^3 B_j \epsi_j + O(\epsi^2).
\end{equations}
We note that $J H(\zeta_\star) J^* = \lambda_n\big( H(\zeta_\star) \big) \cdot \Id_2$ by definition of $J$. 
We write $B_j$ in the basis of Pauli matrices: $B_j = \sum_{k=0}^3 a_{jk} \sigma_k$. This yields
\begin{equations}\label{eq:0i}
J H(\zeta_\star+\epsi) J^* = \left( \lambda_n\big( H(\zeta_\star) \big) +  \sum_{j=1}^3 a_{j0} \epsi_j \right) \cdot \Id_2 +  \sum_{j, k=1}^3 a_{jk} \sigma_k \epsi_j + O(\epsi^2).
\end{equations}
Let $A_\star$ be the $3 \times 3$ matrix with entries $a_{jk}$, $1 \leq j,k \leq 3$. From \S\ref{sec:3.1},  the eigenvalues of $J H(\zeta_\star+\epsi) J^*$ are
\begin{equation}\label{eq:0r}
\lambda_n\big( H(\zeta_\star) \big) + \lr{a_0,\epsi} \pm |A_\star\epsi| + O(\epsi^2).
\end{equation}
On the other hand, the eigenvalues of $J H(\zeta_\star+\epsi) J^*$ are $\lambda_n\big(H(\zeta_\star+\epsi)\big) + O(\epsi^2)$ and $\lambda_{n+1}\big( H(\zeta_\star+\epsi)\big) + O(\epsi^2)$ -- for details, see the proof of \cite[(2.19)]{D19b}. Since these intersect conically, $A_\star$ must be invertible. 

For $\epsi \neq 0$, the matrix $\sum_{j, k=1}^3 a_{jk} \sigma_k \epsi_j$ has two opposite, distinct eigenvalues. Let $\pi_0(\epsi)$ be the projector to the negative eigenvalue. Then
\begin{equations}
\pi(\zeta_\star+\epsi) = \pi_0(\epsi) + O(|\epsi|), \ \ \ \ \nabla\pi(\zeta_\star+\epsi) = \nabla\pi_0(\epsi) + O(1), 
\\
 \nabla\pi(\zeta_\star+\epsi) = O(|\epsi|^{-1}),  \  \ \ \ \nabla\pi_0(\zeta_\star+\epsi) = O(|\epsi|^{-1}).
\end{equations}
We refer to the proof of \cite[Lemma 2.4]{D19b} for such estimates. It follows that
\begin{equation}
	b(\zeta_\star+\epsi) = b_0(\epsi)  + O(|\epsi|^{-1}), \ \  \ \ \text{where} \ \ b_0(\epsi) =   \Tr_{\C^N} \left( \pi_0(\epsi) \big[ \p_{\xi_1} \pi_0(\epsi),  \p_{\xi_2} \pi_0(\epsi) \big]\right).
\end{equation}
Grouping with \eqref{eq:0l}, we obtain $B_n(\zeta_\star+\epsi) = b_0(\epsi)  + O(|\epsi|^{-1})$. In particular,
\begin{equations}\label{eq:0o}
\int_{|\xi-\xi_\star| \leq r} B_n(s_\star\pm\delta,\xi) d\xi =  \int_{|\xi-\xi_\star| \leq r} b_0(\pm\delta,\xi-\xi_\star) d\xi + O\left(\int_{|\xi-\xi_\star| \leq r} \dfrac{1}{|\xi-\xi_\star|} d\xi\right)
\\
= \int_{|\xi| \leq r} b_0(\pm\delta,\xi) d\xi + O(r).
\end{equations}

6. Since $A_\star$ is invertible, the three matrices $A_j = \sum_{k=1}^3 a_{jk} \sigma_k$, $1 \leq j \leq 3$, form a basis of $\EE_0$. We apply the Gran--Schmidt process to $(A_1, A_2, A_3)$: there exists $(\tsigma_1,\tsigma_2,\tsigma_3)$ orthnormal basis of $\EE_0$ and $(t_{jk}) \in M_3(\R)$ upper triangular with positive elements on the diagonal such that $A_j = \sum_{k=1}^3 t_{jk} \tsigma_k$. 

We write $\tsigma_k = \epsilon_\star \cdot U \sigma_k U^*$, where $\epsilon_\star$ is the determinant of $(\tsigma_1,\tsigma_2,\tsigma_3)$ with respect to $(\sigma_1,\sigma_2,\sigma_3)$ -- see \eqref{eq:0h}. In particular, $\epsilon_\star = \sgn\big( \det(A_\star)\big)$. It follows that
\begin{equation}
A_j = \epsilon_\star \cdot U \left(\sum_{k=1}^3 t_{jk} \sigma_k\right) U^*, \ \ \ \ \sum_{j, k=1}^3 a_{jk} \sigma_k \epsi_j = \epsilon_\star \cdot U \left(\sum_{j, k=1}^3 t_{jk} \sigma_k \epsi_j \right) U^*.
\end{equation}
Hence, $\pi_0(\epsi)$ is, up to conjugation, the projector associated to the negative eigenvalue of $\epsilon_\star \cdot \sum_{j, k=1}^3 t_{jk} \sigma_k \epsi_j$. 

We define more appropriate coordinates
\begin{equation}\label{eq:0k}
\txi_1 = \dfrac{t_{12}\delta + t_{22}\xi_1}{t_{11}\delta}, \ \ \ \ \txi_2 = \dfrac{t_{13}\delta + t_{23}\xi_1 + t_{33}\xi_2}{t_{11}\delta}.
\end{equation}
Using invariance of two-forms under change of coordinates, $b_0(\xi)d\xi = \tb_0(\txi)d\txi$, where $\tb_0(\pm \delta,\txi) d\xi$ is the two-form associated to the negative eigenspace of
$\epsilon_\star \delta \cdot t_{11} \big( \sigma_1 + \sigma_2 \txi_1 + \sigma_3 \txi_2 \big)$. This setup allows us to apply \cite[(23)]{FC13}, which gives:
\begin{equation}
\tb_0(\pm \delta,\txi) = \dfrac{i\epsilon_\star^3 (\pm\delta)^3}{2\delta^3\big(\txi^2_1 + \txi_2^2 + 1\big)^{3/2}} \matrice{1 \\ \txi_1\\ \txi_2  } \cdot \matrice{0 \\ 1 \\ 0} \wedge \matrice{0 \\ 0 \\ 1} = \dfrac{\pm\epsilon_\star}{2\big(\txi^2_1 + \txi_2^2 + 1\big)^{3/2}}.
\end{equation}
Under the change of coordinates \eqref{eq:0k}, the disk $|\xi| \leq r$ gets mapped to an ellipse centered at distance $O(1)$ from the origin, of dimensions $\sim \delta^{-1}$. Thus,
\begin{equation}\label{eq:0n}
\int_{|\xi| \leq r} b_0(\pm\delta,\xi) d\xi = \int_{\R^2} b_0(\pm\delta,\xi) d\xi + O(\delta) = \pm \epsilon_\star \pi + O(\delta).
\end{equation}
We refer to the proof of \cite[Lemma 2.5]{D19b} for details.

7. Grouping  \eqref{eq:0d}, \eqref{eq:0e}, \eqref{eq:0o} and \eqref{eq:0n}, we end up with
\begin{equation}
c_1(\VVV_1)-c_1(\VVV_0) = \sum_{\zeta_\star \in \RRR} \epsilon_\star + O(r+\delta) = \sum_{\zeta_\star \in \RRR} \sgn\big(\det(A_\star)\big) + O(r+\delta).
\end{equation}
Making $\delta \rightarrow 0$, we end up with
\begin{equation}\label{eq:0q}
c_1(\VVV_1)-c_1(\VVV_0) = \sum_{\zeta_\star \in \RRR} \sgn\big(\det(A_\star)\big) + O(r).
\end{equation}
 Taking $r$ sufficiently small, the term $O(r)$ is at most $1/2$. Since both sides of \eqref{eq:0q} are integers, we end up with 
\begin{equation}
c_1(\VVV_1)-c_1(\VVV_0)  = \sum_{\zeta_\star \in \RRR} \sgn\big(\det(A_\star)\big).
\end{equation}
This completes the proof. \end{proof}

\appendix

\section{Continuous approximation}

Let $H \in \LL$ with a conical degeneracy at $(s_0, \xi_0) \in (0,1) \times \Tt^2$ and $\HH^\delta$ defined as in \eqref{eq:3a}. In this appendix, we derive formally the effective Dirac equation \eqref{eq:0p}. It describes the evolution of amplitudes to solutions of $(D_t-\HH^\delta) \psi = 0$ that are initially concentrated (in phase-space) near $(\R e_1+s_0e_2,\xi_0)$.

\subsection{Reduction to $(s_0, \xi_0) = (0,0)$}\label{sec:6.1}  We show that $\HH^\delta$ is unitarily equivalent to an operator with a conical degeneracy at $(0,0)$. Define
\begin{equation}\label{eq:3f}
\tH_s(\xi) = H_{s_0+s}(\xi+\xi_0), \ \ \ \ \tH_{s-s_0}(\xi-\xi_0) = H_s(\xi),
\end{equation}
and $\tHH_s$, $\tHH^\delta$ relative to $\tH$, according to \eqref{eq:3b} and \eqref{eq:3a}. 

For $m \in \Z^2$, set $\ell = m - \big[\delta^{-1} s_0\big] e_2$, where $\big[\delta^{-1} s_0\big]$ stands for the integer part of $\delta^{-1} s_0$.
For $\phi \in \ell^2(\Z^2,\C^N)$, we have:
\begin{equations}
e^{i\xi_0 \ell } \cdot \big( \tHH^\delta \phi \big)(\ell) 
= e^{i\xi_0 \ell } \cdot \big( \tHH_{\delta n_2-s_0}  \phi \big)(\ell) 
\\
 =  \int_{\Tt^2} e^{i(\xi+\xi_0) \ell} \cdot \tH_{\delta n_2-s_0}(\xi) \hphi(\xi) \cdot \dfrac{d\xi}{(2\pi)^2} 
 =  \int_{\Tt^2} e^{i\xi \ell} \cdot \tH_{\delta n_2-s_0}(\xi-\xi_0) \hphi(\xi-\xi_0) \cdot \dfrac{d\xi}{(2\pi)^2}
 \\
 = \int_{\Tt^2} e^{i\xi \ell} \cdot H_{\delta n_2}(\xi) \widehat{e^{i\xi_0 \cdot}\phi}(\xi) \cdot \dfrac{d\xi}{(2\pi)^2} = \left(\HH^\delta e^{i\xi_0 \cdot} \phi\right) (\ell).
\end{equations}
This means that $U \tHH^\delta U^* = \HH^\delta$, where
\begin{equation}
U \phi(m) = \left(e^{i\xi_0 \cdot} \phi\right)\big( m - [\delta^{-1}s_0] e_2\big), \ \ \ \ \ U^* \phi(m) = e^{i\xi_0 m} \cdot \phi\big( m + [\delta^{-1}s_0] e_2 \big).
\end{equation}

\subsection{Effective equation} Since $H_s(\xi)$ has a conical degeneracy at $(s_0,\xi_0)$, there exists $f_1, f_2 \in \C^N$  satisfying \eqref{eq:1g}. As $\delta \rightarrow 0$, we derive (formally) the leading asymptotics of $\HH^\delta \phi$, where 
\begin{equation}\label{eq:3g}
\phi(m) = e^{i\xi_0 m} \cdot \sum_{j=1}^2 \az_j(s_0 e_2+\delta^{1/2} m) f_j \ \in \ \ell^2(\Z^2,\C^N), \ \ \ \ \ \az \in C_0^\infty(\R^2,\C^N).
\end{equation}

After rescaling, $\phi$ is semiclassically (scale $\delta$) localized near $(\R e_1 + s_0e_2,\xi_0)$. We write \eqref{eq:3g} as $\phi \simeq U^* \vp$, where $\vp(m) = \sum_{j=1}^2 \az_j(\delta^{1/2} m) f_j = J^* \az(\delta^{1/2} m)$, 
and $J : \C^N \rightarrow \C^2$ is the operator of \eqref{eq:1g}. Using a Riemann sum argument, we observe that as $\delta \rightarrow 0$,
\begin{equation}\label{eq:3j}
\hspace{-3mm}
\delta \cdot \hvp(\delta^{1/2} \xi) = J^* \left( \delta \sum_{m \in \Z^2} e^{-i \delta^{1/2} \xi m} \az(\delta^{1/2} m) \right) \simeq J^* \left( \int_{\R^2} e^{-i \xi x} \az(x) dx\right) = J^* \haz(\xi).
\end{equation}

Thanks to \S\ref{sec:6.1}, we have $\HH^\delta \phi = \HH^\delta U^* \vp = U^* \tHH^\delta \vp$. Now, we compute $\tHH^\delta \vp$:
\begin{equations}\label{eq:3h}
\left(\tHH^\delta \vp \right)(m) =  \left(\tHH_{\delta m_2} \vp\right)(m)   =  \int_{\Tt^2} e^{i\xi m} \cdot \tH_{\delta m_2}(\xi)  \hvp(\xi) \cdot \dfrac{d\xi}{(2\pi)^2}
\\
=  \int_{\delta^{-1} \Tt^2} e^{i\delta^{1/2} \xi m} \cdot H_{\delta m_2}(\delta^{1/2} \xi)  \hvp(\delta^{1/2} \xi) \cdot \dfrac{\delta d\xi}{(2\pi)^2},
\end{equations}
where we made the substitution $\xi \mapsto \delta^{1/2} \xi$. Since $\vp$ is spectrally concentrated near $0$, it is reasonable to replace the integration domain in \eqref{eq:3h} to $\R^2$. Using \eqref{eq:3j}, we get
\begin{equation}\label{eq:3m}
\left(\tHH^\delta \vp \right)(m) \simeq  \int_{\R^2} e^{i\delta^{1/2} \xi m} \cdot H_{\delta m_2}(\delta^{1/2} \xi) J^* \haz(\xi) \cdot \dfrac{d\xi}{(2\pi )^2}.
\end{equation}

The identity \eqref{eq:1g} allows us to expand $H_{\delta m_2}(\delta^{1/2} \xi) J^*$ as
\begin{equation}\label{eq:3l}
H_{\delta m_2}(\delta^{1/2} \xi) J^* \simeq J^* \big( E_0 + \delta^{1/2} \cdot \Di(\delta^{1/2} m_2, \xi) \big), 
\end{equation}
where $\Di(s,\xi)$ is a family of $2 \times 2$ matrices depending linearly on $(s,\xi)$, and $E_0$ is the energy of the conical crossing. Plugging \eqref{eq:3l} into \eqref{eq:3m}, we obtain
\begin{equation}
\big(\tHH^\delta \vp \big)(m) \simeq J^*\int_{\R^2} e^{i\delta^{1/2} \xi m} \cdot \big( E_0 + \delta^{1/2} \Di(\delta^{1/2} m_2, \xi) \big) \haz(\xi) \cdot \dfrac{d\xi}{(2\pi )^2} = J^* \big( E_0 + \delta^{1/2} \Di\big) \az (\delta^{1/2} m),
\end{equation}
where $\Di = \Di(x_2, D_x)$ is a Dirac operator. Since $\vp(m) = J^*\az(\delta^{1/2}m)$, this means that $J^*$ approximately intertwines between $\tHH^\delta$ and $E_0 + \delta^{1/2} \Di$, for adiabatic data. 

Up to a phase and a time-rescaling, the equations $D_t - E_0 - \delta^{1/2} \Di$ and $D_t - \Di$ are equivalent. Using the above intertwining, we conclude that 
 $(D_t-\HH^\delta) \psi = 0$ has approximate solutions whose asymptotics are slow linear combinations of $f_1$ and $f_2$:
\begin{equation}
e^{i(E_0 t + \xi_0 m)} \cdot \sum_{j=1}^2 \beta_j\left( \delta^{1/2} t, \ s_0 e_2 + \delta^{1/2} m \right) f_j,
\end{equation}
with amplitudes $\beta_j(t,x)$ solving the Dirac equation \eqref{eq:0p}: $(D_t - \Di) \beta = 0$.

\end{document}